\documentclass[acmsmall,nonacm]{acmart}

\setcopyright{none} 
\setcctype{by}
\acmDOI{10.1145/3808328}
\acmYear{2026}
\acmJournal{PACMPL}
\acmVolume{10}
\acmNumber{PLDI}
\acmArticle{250}
\acmMonth{6}
\acmSubmissionID{pldi26main-p595-p}
\received{2025-11-14}
\received[accepted]{2026-04-03}

\newcommand\appref[1]{App.~\ref{#1}}

\usepackage{subcaption} 

\usepackage{prftree} 
\usepackage{tikz}
\usetikzlibrary{shapes, positioning, arrows, positioning, shapes.multipart, calc, decorations.pathmorphing, fit, backgrounds}
\usepackage{bm}
\hypersetup{hypertexnames=false} 

\DeclareMathOperator*{\mcup}{\text{\raisebox{0.25ex}{\scalebox{0.8}{$\bigcup$}}}}
\DeclareMathOperator*{\motimes}{\text{\raisebox{0.25ex}{\scalebox{0.8}{$\bigotimes$}}}}

\newcommand{\const}[1]{\mathsf{#1}}
\newcommand{\hasrt}[1]{\blue{[#1]}}
\newcommand{\hasrtl}[1]{\blue{[#1}}
\newcommand{\hasrtr}[1]{\blue{#1]}}
\newcommand{\while}[2]{\resv{while}\ #1\ \resv{do}\ #2\ \resv{od}}
\NewDocumentCommand{\Mono}{ooo}{
	\IfNoValueTF{#1}
	{\const{mono}} 
	{\const{mono}(#1,#2,#3)} 
}
\NewDocumentCommand{\hAdd}{ooo}{
	\IfNoValueTF{#1}
	{\oplus} 
	{#1 = #2 \oplus #3} 
}
\NewDocumentCommand{\hMonadAdd}{ooo}{
	\IfNoValueTF{#1}
	{\oplus} 
	{\IfNoValueTF{#3}
	{#1 \oplus #2} 
	{#1 = #2 \oplus #3} 
	}
}
\NewDocumentCommand{\addStates}{ooo}{
	\IfNoValueTF{#1}
	{\oplus} 
	{#1 = #2 \oplus #3} 
}
\NewDocumentCommand{\plw}{ooo}{
	\IfNoValueTF{#1}
	{\const{plw}} 
	{\const{plw}(#1,#2,#3)} 
}
\NewDocumentCommand{\prw}{ooo}{
	\IfNoValueTF{#1}
	{\const{prw}} 
	{\const{prw}(#1,#2,#3)} 
}
\NewDocumentCommand{\pw}{ooo}{
	\IfNoValueTF{#1}
	{\const{pw}} 
	{\const{pw}(#1,#2,#3)} 
}
\NewDocumentCommand{\Res}{oo}{
	\IfNoValueTF{#1}
	{\sqsupseteq} 
	{#1\sqsupseteq#2} 
}
\NewDocumentCommand{\Sem}{oo}{
	\IfNoValueTF{#1}
	{\const{sem}} 
	{\const{sem}(#1,#2)} 
}
\NewDocumentCommand{\md}{o}{
	\IfNoValueTF{#1}
	{\const{md}} 
	{\const{md}(#1)} 
}
\NewDocumentCommand{\fv}{o}{
	\IfNoValueTF{#1}
	{\const{fv}} 
	{\const{fv}(#1)} 
}
\NewDocumentCommand{\lowa}{oo}{ 
	\IfNoValueTF{#1}
	{\const{low}} 
	{\IfNoValueTF{#2}
	 {\const{low}(#1)} 
	 {\const{low}(#1,#2)} 
	}
}
\NewDocumentCommand{\List}{oo}{
	\IfNoValueTF{#1}
	{\const{list}} 
	{\const{list}(#1,#2)} 
}
\NewDocumentCommand{\Lseg}{ooo}{
	\IfNoValueTF{#1}
	{\const{lseg}} 
	{\const{lseg}(#1,#2,#3)} 
}
\NewDocumentCommand{\len}{o}{
	\IfNoValueTF{#1}
	{\const{len}} 
	{\const{len}(#1)} 
}
\NewDocumentCommand{\Rel}{ooo}{
	\IfNoValueTF{#1}
	{\const{R}} 
	{\const{R}_{#1}(#2,#3)} 
}
\NewDocumentCommand{\F}{o}{
	\IfNoValueTF{#1}
	{\mathcal{F}} 
	{\mathcal{F}(#1)} 
}
\NewDocumentCommand{\Pow}{o}{
	\IfNoValueTF{#1}
	{\mathcal{P}} 
	{\mathcal{P}(#1)} 
}
\NewDocumentCommand{\R}{o}{
	\IfNoValueTF{#1}
	{\mathcal{R}} 
	{\mathcal{R}(#1)} 
}
\NewDocumentCommand{\baseok}{o}{
	\IfNoValueTF{#1}
	{\text{no }\sigma(\er:\sth)} 
	{\text{no }\sigma(\er:\sth)\text{ in }#1} 
}
\newcommand{\resv}[1]{\textnormal{\textbf{#1}}}
\newcommand{\la}{\langle}
\newcommand{\ra}{\rangle}
\newcommand{\tb}[1]{\la #1\ra} 
\newcommand{\fs}[1]{\forall\tb{#1}\ldotp} 
\newcommand{\es}[1]{\exists\tb{#1}\ldotp} 
\newcommand{\fn}[1]{\forall #1\ldotp} 
\newcommand{\en}[1]{\exists #1\ldotp} 
\newcommand{\sth}{\_}
\newcommand{\Skip}{\resv{skip}}
\newcommand{\Assign}[2]{#1\coloneqq#2}
\newcommand{\Havoc}[1]{#1\coloneqq\mathrm{nonDet}()}
\newcommand{\Assume}[1]{\resv{assume}\ #1}
\newcommand{\Assert}[1]{\resv{assert}\ #1}
\newcommand{\Alloc}[1]{#1\coloneqq\mathrm{alloc}()}
\newcommand{\Write}[2]{[#1]\coloneqq#2}
\newcommand{\Read}[2]{#1\coloneqq[#2]}
\newcommand{\Free}[1]{\mathrm{free}(#1)}
\newcommand{\If}[3]{\resv{if}\ #1\ \resv{then}\ #2\ \resv{else}\ #3\ \resv{fi} }
\newcommand{\While}[2]{\resv{while}\ #1\ \resv{do}\ #2\ \resv{od}}
\newcommand{\cHavoc}[3]{#1\coloneqq\mathrm{randInt}(#2,#3)}
\newcommand{\cEr}{\resv{Aborts}\erl}
\newcommand{\cDv}{\resv{Aborts}\dvl}
\newcommand{\okl}{_{\mathrm{ok}}}
\newcommand{\erl}{_{\mathrm{er}}}
\newcommand{\ukl}{_{\mathrm{uk}}}
\newcommand{\dvl}{_{\mathrm{dv}}}
\newcommand{\ok}{{\mathrm{ok}}}
\newcommand{\er}{{\mathrm{er}}}
\newcommand{\uk}{{\mathrm{uk}}}
\newcommand{\el}{_{\epsilon}}
\newcommand{\elp}{_{\epsilon'}}
\newcommand{\PVars}{\const{PVars}}
\newcommand{\GVars}{\const{DVars}}
\newcommand{\SVars}{\const{SVars}}
\newcommand{\defiff}{\overset{\text{def}}{\iff}}
\newcommand{\emp}{\mathrm{emp}}
\newcommand{\UNIV}{\const{UNIV}}

\newtheorem{lemma}{Lemma}
\newtheorem{theorem}{Theorem}

\usepackage{listings}

\definecolor{darkgreen}{rgb}{0.0,0.6,0.0}

\newcommand\secref[1]{Sect.~\ref{#1}}
\newcommand\figref[1]{Fig.~\ref{#1}}
\newcommand\thmref[1]{Thm.~\ref{#1}}

\newcommand{\wrt}{{{w.r.t.\@}}}

\newcommand{\eg}{{{e.g.,~}}}
\newcommand{\ie}{{{i.e.,~}}}

\newcommand{\greenify}[1]{\textcolor{darkgreen}{#1}}

\usepackage{mathtools}

\makeatletter

\newskip \point

\def \premisespacing{\quad}

\def \RulePremisesNewlineMore[#1]#2.#3#4{\@ifnextchar\bgroup{\RulePremisesNewlineMore[#1]{#2}.{#3\premisespacing#4}}{\@ifnextchar.{\RulePremisesNewline[#1]{{\begin{array}{c}#2\\#3\premisespacing#4\end{array}}}}{\RuleMultiPremise[#1]{{\begin{array}{c}#2\\#3\end{array}}}{#4}}}}

\def \RulePremisesNewline[#1]#2.#3{\@ifnextchar\bgroup{\RulePremisesNewlineMore[#1]{#2}.{#3}}{\@ifnextchar.{\RulePremisesNewline[#1]{{\begin{array}{c}#2\\#3\end{array}}}}{\RuleMultiPremise[#1]{#2}{#3}}}}

\def \RuleMultiPremise[#1]#2#3{\@ifnextchar\bgroup{\RuleMultiPremise[#1]{#2\premisespacing#3}}{\@ifnextchar.{\RulePremisesNewline[#1]{#2\premisespacing#3}}{\prooftree #2\justifies#3 \using{#1}\endprooftree}}}

\def \RuleWithName[#1]#2{\@ifnextchar\bgroup {\RuleMultiPremise[#1]{#2}}{\@ifnextchar.{\RulePremisesNewline[#1]{#2}}{\prooftree \justifies #2 \using{#1} \endprooftree}}}

\def \RuleWithInfo[#1]{\@ifnextchar[{\RuleWithNameAndCondition[#1]}{\RuleWithName[(#1)]}}

\def \RuleWithNameAndCondition[#1][#2]{\RuleWithName[(#1)^{#2}]}

\def \Inf{\proofrulebaseline=2ex \abovedisplayskip12\point\belowdisplayskip12\point \abovedisplayshortskip8\point\belowdisplayshortskip8\point \@ifnextchar[{\RuleWithInfo}{\RuleWithName[ ]}}

\makeatother

\newcommand{\simpleHoare}[3]{\ensuremath{ \blue{ [ #1 ] } \; #2 \; \blue{ [ #3 ] } }}

\newcommand{\largeHoare}[3]{
    \begin{center}
        $\blue{\left[ #1 \right]}$ \\
        $#2$ \\
        $\blue{\left[ #3 \right]}$
    \end{center}
}

\newcommand{\hoare}[3]{\ensuremath{{\models} \simpleHoare{#1}{#2}{#3} }}
\newcommand{\nothoare}[3]{\ensuremath{{\centernot{\models}} \simpleHoare{#1}{#2}{#3} }}
\newcommand{\shoare}[3]{\ensuremath{{\vdash} \simpleHoare{#1}{#2}{#3} }}

\newcommand{\normalHoare}[3]{\ensuremath{\brown{\{ #1 \}} \; #2 \; \brown{\{ #3 \}} }}

\newcommand{\update}[3]{\ensuremath{#1[#2 \mapsto #3]}}

\newcommand{\acc}[1]{\ensuremath{#1 \mapsto \_}}
\newcommand{\pointsto}[2]{\ensuremath{#1 \mapsto #2}}
\newcommand{\pointstobot}[1]{\ensuremath{#1 \mapsto \bot}}

\newcommand{\cseq}{ \mathbf{;} \; }
\newcommand{\cassign}[2]{\mathit{#1} \coloneqq \mathit{#2}}

\newcommand{\cfree}[1]{\ensuremath{\mathbf{free}(#1)}}
\newcommand{\calloc}[1]{\cassign{#1}{\mathbf{alloc}()}}
\newcommand{\cread}[2]{\cassign{#1}{[#2]}}

\definecolor{royalblue}{rgb}{0.25, 0.41, 0.88}
\newcommand{\blue}[1]{{\color{royalblue}{#1}}}

\definecolor{somebrown}{rgb}{0.8, 0.35, 0.1}
\newcommand{\brown}[1]{{\color{somebrown}{#1}}}

\usepackage{pftools}
\usepackage{mathpartir}

\newcommand{\applyruleOld}[3]{\ensuremath{\infer*[right=\ref{#1}]{#2}{#3}}}

\NewDocumentCommand{\applyrule}{m O{} O{} O{} O{} O{} m}{
  \IfBlankTF{#2}{
    \applyruleOld{#1}{}{#7}
  }{
    \IfBlankTF{#3}{
      \applyruleOld{#1}{#2}{#7}
    }{
      \IfBlankTF{#4}{
        \applyruleOld{#1}{#2 \\ #3}{#7}
      }{
        \IfBlankTF{#5}{
          \applyruleOld{#1}{#2 \\ #3 \\ #4}{#7}
        }{
            \IfBlankTF{#6}{
                \applyruleOld{#1}{#2 \\ #3 \\ #4 \\ #5}{#7}
            }{
                \applyruleOld{#1}{#2 \\ #3 \\ #4 \\ #5 \\ #6}{#7}
            }
        }
      }
    }
  }
}

\NewDocumentCommand{\newrule}{m m O{} O{} O{} O{} O{} m}{
  \IfBlankTF{#3}{
    \inferhref{#1}{#2}{}{#8}
  }{
    \IfBlankTF{#4}{
      \inferhref{#1}{#2}{#3}{#8}
    }{
      \IfBlankTF{#5}{
        \inferhref{#1}{#2}{#3 \\ #4}{#8}
      }{
        \IfBlankTF{#6}{
          \inferhref{#1}{#2}{#3 \\ #4 \\ #5}{#8}
        }{
            \IfBlankTF{#7}{
                \inferhref{#1}{#2}{#3 \\ #4 \\ #5 \\ #6}{#8}
            }{
                \inferhref{#1}{#2}{#3 \\ #4 \\ #5 \\ #6 \\ #7}{#8}
            }
        }
      }
    }
  }
}

\makeatletter
\newcommand{\leqnomode}{\tagsleft@true}
\newcommand{\reqnomode}{\tagsleft@false}
\makeatother

\usepackage{marvosym} 

\begin{document}

\title{Hyper Separation Logic}

\author{Trayan Gospodinov}
\orcid{0009-0007-6169-7939}
\affiliation{%
	\institution{INSAIT, Sofia University “St. Kliment Ohridski”}
	\city{Sofia}
	\country{Bulgaria}
}
\email{trayan.gospodinov@insait.ai}

\author{Peter Müller}
\orcid{0000-0001-7001-2566}
\affiliation{%
	\institution{ETH Zurich}
	\city{Zurich}
	\country{Switzerland}
}
\email{peter.mueller@inf.ethz.ch}

\author{Thibault Dardinier}
\orcid{0000-0003-2719-4856}
\affiliation{%
	\institution{ETH Zurich}
	\city{Zurich}
	\country{Switzerland}
}
\email{thibault.dardinier@inf.ethz.ch}

\begin{abstract}
Many important functional and security properties---including non-interference, determinism, and generalized non-interference (GNI)---are hyperproperties, i.e., properties relating multiple executions of a program. Existing separation logics allow one to reason about specific classes of hyperproperties, e.g., $\forall\forall$-hyperproperties such as non-interference and $\exists\exists$-properties such as non-determinism. However, they do not support quantifier alternation, which is for instance needed to express GNI. The only existing logic that can reason about such properties is Hyper Hoare Logic, but it does not support heap-manipulating programs and, thus, is not applicable to common imperative programs.

This paper introduces Hyper Separation Logic (HSL), the first program logic that supports modular reasoning about hyperproperties with arbitrary quantifier alternation over programs that manipulate the heap. HSL generalizes Hyper Hoare Logic with a novel hyper separating conjunction that lifts the standard separating conjunction to sets of states, enabling a generalized frame rule for hyperproperties. We prove HSL sound in Isabelle/HOL and demonstrate its expressiveness for hyperproperties that lie beyond the reach of existing separation logics.
\end{abstract}



\maketitle

\section{Introduction}
\label{sec:introduction}
Many important functional and security properties are \emph{hyperproperties}~\cite{clarksonHyperproperties2008},
\ie properties that relate \emph{multiple} executions of a program.
For example, information flow security, which requires that confidential inputs do not leak through public outputs,
can be formalized as \emph{non-interference} (NI)~\cite{goguenSecurityPoliciesSecurity1982}:
NI requires that any two executions with the same public inputs (but potentially different confidential inputs)
produce the same public outputs.
NI is a $\forall\forall$-hyperproperty (also called \emph{2-hypersafety}), \ie a property of all \emph{pairs} of executions.
Important functional hyperproperties
include
determinism ($\forall \forall$),
transitivity ($\forall \forall \forall$, important for custom comparators~\cite{sousaCartesianHoareLogic2016}),
or associativity ($\forall \forall \forall \forall$, important for parallel reductions).

Many hyperproperties, however, fall outside the class of $\forall^k$-hyperproperties because they require existential quantification over executions.
For example, reachability is an $\exists$-hyperproperty, as it asserts the existence of an execution
that reaches a given state, while non-determinism is an $\exists\exists$-hyperproperty, requiring two
executions with the same inputs that produce different outputs.
More expressive hyperproperties require \emph{quantifier alternation} (\ie $\forall^* \exists^*$ or $\exists^* \forall^*$).
An important example is \emph{generalized non-interference} (GNI)~\cite{mcculloughNoninterferenceComposabilitySecurity1988},
a weakening of NI for non-deterministic programs.
GNI is a $\forall\forall\exists$-hyperproperty: it permits two executions $\tau_1$ and $\tau_2$ with the same low inputs
to have different low outputs, provided that there exists a third execution $\tau$ with the same low inputs,
the same high inputs as $\tau_1$, and the same low outputs as $\tau_2$.
Hyperproperties with $\exists^* \forall^*$ prefixes are also important, for example to express the existence of a minimum
or a violation of GNI.

As reasoning about hyperproperties has gained importance,
many Hoare-like program logics for hyperproperties have been proposed.
Relational Hoare Logic~\cite{bentonSimpleRelationalCorrectness2004} extends Hoare Logic~\cite{floydAssigningMeaningsPrograms1967, hoareAxiomaticBasisComputer1969} to $\forall \forall$-hyperproperties, and has later been extended to $\forall^k$-hyperproperties~\cite{sousaCartesianHoareLogic2016,dosualdoProvingHypersafetyCompositionally2022}.
A complementary line of work has developed program logics for $\exists$ \cite{devriesReverseHoareLogic2011,ohearnIncorrectnessLogic2019} and $\exists \exists$~\cite{murrayUnderApproximateRelationalLogic2020} hyperproperties,
targeting reachability and incorrectness.
Recent work has either unified safety and incorrectness reasoning in a single logic~\cite{zilbersteinOutcomeLogicUnifying2023},
or has developed logics for $\forall^* \exists^*$-hyperproperties~\cite{dickersonRHLEModularDeductive2022,antonopoulosAlgebraAlignmentRelational2023,beutnerAutomatedSoftwareVerification2024} (such as GNI).
The only program logic that handles arbitrary quantifier alternation is Hyper Hoare Logic (HHL)~\cite{dardinierHyperHoareLogic2024}.

None of these logics apply to heap-manipulating programs.
Separation Logic (SL)~\cite{reynoldsSeparationLogicLogic2002} supports reasoning about such programs
in a modular, local way via its \emph{frame rule}.
Thus, unsurprinsingly, many Hoare logics for hyperproperties have been extended to support SL reasoning.
Separation logics for $\forall^k$-hyperproperties include Relational Separation Logic~\cite{yangRelationalSeparationLogic2007}
(for $k = 2$), and LGTM~\cite{gladshteinMechanisedHypersafetyProofs2024} (for unbounded $k$),
while separation logics for $\exists^k$-hyperproperties include Incorrectness Separation Logic~\cite{raadLocalReasoningPresence2020,raadConcurrentIncorrectnessSeparation2022}
and Sufficient Incorrectness Separation Logic~\cite{ascariSufficientIncorrectnessLogic2024}
(for $k = 1$), and InsecSL~\cite{murrayCompositionalVulnerabilityDetection2023} (for $k = 2$).
Additionally, Outcome Separation Logic~\cite{zilbersteinOutcomeSeparationLogic2024} and Exact Separation Logic~\cite{maksimovicExactSeparationLogic2023}
unify safety and reachability reasoning for \emph{single} executions, \ie they support a mix of $\forall$ and $\exists$ properties.

However, no existing program logic supports reasoning about
$\forall^+ \exists^+$, $\exists^+ \forall^+$, or even $\exists^k$ (for $k > 2$)
hyperproperties for heap-manipulating programs. Without a logic that supports such properties, it is extremely difficult to prove useful properties such as GNI for realistic imperative programs.

\paragraph{This work}
We present \emph{Hyper Separation Logic} (HSL), the first program logic that supports proving hyperproperties with arbitrary quantifier alternation for heap-manipulating programs.
HSL builds on the foundations of HHL:
It establishes \emph{hyper-triples} of the form $\simpleHoare{P}{C}{Q}$,
where $P$ and $Q$ are \emph{hyper-assertions}, \ie arbitrarily-quantified predicates over sets of states (including a heap).
HSL supports local, modular reasoning via a novel \emph{generalized frame rule} for hyper-assertions:
From a local triple $\simpleHoare{P}{C}{Q}$,
one can derive a new triple $\simpleHoare{P \star F}{C}{Q \star F}$ (under suitable side conditions),
which can be applied in a larger context.
Here, $\star$ is a novel \emph{hyper separating conjunction},
which generalizes the separating conjunction $*$ from SL\@.
Intuitively, the hyper-assertion $P \star F$ expresses that a set $S$ of states can be, per state, split into two sets of states $S_P$ (satisfying $P$) and $S_F$ (satisfying $F$), such that each state in $S$ is a disjoint union of a state in $S_P$ and a state in $S_F$.

\paragraph{Contributions}
Our main contributions are:
\begin{itemize}
	\item We present \emph{Hyper Separation Logic} (HSL), the first separation logic for hyperproperties with arbitrary quantifier alternation.
	\item We introduce a novel \emph{hyper separating conjunction} connective between hyper-assertions, which preserves upper bounds (for universally quantified executions) and lower bounds (for existentially quantified executions) of sets of states.
	\item We introduce a novel definition of hyper-triple validity, which guarantees soundness of the generalized frame rule by construction. Our definition is based on a novel decomposition of $\forall^*$-hyperproperties into $\forall$-properties.
	\item To support reachability reasoning, we introduce a novel notion of \emph{unknown} states,
	and weaken the definition of hyper-triple validity accordingly.
	\item We introduce and prove sound a proof system for HSL, including rules for basic heap-manipulating commands, a generalized frame rule, and rules to reason about errors.
	\item To ease reasoning, we introduce a syntax for hyper-assertions, as well as a novel notion of \emph{scaffold} variables, and use the latter to derive a more expressive rule to read from the heap. 
	\item Finally, we showcase our logic on a diverse set of examples, including for $\forall^* \exists^*$ and $\exists^* \forall^*$ hyperproperties, which are beyond the reach of existing separation logics.
\end{itemize}

\textbf{All formal results about HSL presented in this paper have been formalized in Isabelle/HOL~\cite{nipkowIsabelleHOLProof2002},
and our mechanization is publicly available \citep{HyperSeparationLogic_artifact}.}

\paragraph{Outline}
\secref{sec:tour} gives a high-level overview of some of the main novelties of HSL, in particular its hyper-triples, hyper separating conjunction, and generalized frame rule.
\secref{sec:key-ideas} then presents our solutions to key technical challenges,
including how we define the hyper separating conjunction, how we ensure the soundness of the frame rule with a novel decomposition of $\forall^*$-hyperproperties,
and how our unknown states enable reachability reasoning.
\secref{sec:Hyper-Separation-Logic} presents HSL formally,
while we discuss related work in \secref{sec:related}, and conclude in \secref{sec:conclusion}.

\section{A Tour of Hyper Separation Logic}
\label{sec:tour}

\newcommand{\fsok}[1]{\ensuremath{%
(\fs{\sigma} \sigma(#1))
}}

\newcommand{\lowOne}[1]{\ensuremath{%
(\fs{\sigma_1} \fs{\sigma_2} \sigma_1(#1) = \sigma_2(#1))
}}

\newcommand{\lowTwo}[2]{\ensuremath{%
(\fs{\sigma_1} \fs{\sigma_2} \sigma_1(#1) = \sigma_2(#1) \land \sigma_1(#2) = \sigma_2(#2))
}}

\newcommand{\pureGNI}[2]{\ensuremath{%
(\fs{\sigma_1} \fs{\sigma_2} \es{\sigma}
\sigma(#1) = \sigma_1(#1) \land \sigma(#2) = \sigma_2(#2))
}}

\newcommand{\pureGNIDouble}[3]{\ensuremath{%
(\fs{\sigma_1} \fs{\sigma_2} \es{\sigma}
\sigma(#1) = \sigma_1(#1) \land \sigma(#2) = \sigma_2(#2) \land \sigma(#3) = \sigma_2(#3))
}}

\newcommand{\LowOne}[2]{\ensuremath{%
\fs{\sigma_1} \fs{\sigma_2} 
\exists #2 \ldotp
\sigma_1(#1 \mapsto #2) \land \sigma_2(#1 \mapsto #2)
}}

\newcommand{\LowTwo}[4]{\ensuremath{%
\fs{\sigma_1} \fs{\sigma_2}
\exists #2, #4 \ldotp
\sigma_1(#1 \mapsto #2 * #3 \mapsto #4) \land \sigma_2(#1 \mapsto #2 * #3 \mapsto #4)
}}


\newcommand{\precut}{\ensuremath{%
\LowOne{l}{u}
}}

\newcommand{\frameA}{\ensuremath{%
\fsok{\acc{h}}
}}

\newcommand{\postcut}{\ensuremath{%
\LowTwo{x}{v}{y}{w}
}}

\newcommand{\intermediateAssertion}{\ensuremath{%
\fs{\sigma_1} \fs{\sigma_2}
\exists v, w \ldotp
\sigma_1(x \mapsto v * y \mapsto w * \acc{h}) \land \sigma_2(x \mapsto v * y \mapsto w * \acc{h})
}}

\newcommand{\frameBOld}{\ensuremath{%
\fsok{y \mapsto \delta_y}
\land
\lowOne{\delta_y}
}}

\newcommand{\frameB}{\ensuremath{%
\LowOne{y}{w}
}}

\newcommand{\preComp}{\ensuremath{%
\fs{\sigma_1} \fs{\sigma_2}
\exists v \ldotp
\sigma_1(x \mapsto v * \acc{h}) \land \sigma_2(x \mapsto v * \acc{h})
}}

\newcommand{\postComp}{\ensuremath{%
\fs{\sigma_1} \fs{\sigma_2} \es{\sigma}
\exists u, v \ldotp
\sigma_1(\acc{o} * \pointsto{h}{v})
\land
\sigma(\pointsto{o}{u} * \pointsto{h}{v})
\land
\sigma_2(\pointsto{o}{u} * \acc{h})
}}

\newcommand{\fullPre}{\ensuremath{
\fs{\sigma_1} \fs{\sigma_2}
\exists u \ldotp
\sigma_1(l \mapsto u * \acc{h}) \land \sigma_2(l \mapsto u * \acc{h})
}}

\newcommand{\fullPost}{\ensuremath{
\fs{\sigma_1} \fs{\sigma_2} \es{\sigma}
\exists u, v \ldotp
\sigma_1(\pointsto{h}{u})
\land
\sigma(\pointsto{h}{u} \land r = v)
\land
\sigma_2(\acc{h} \land r = v)
}}

\newcommand{\fullPostBefore}{\ensuremath{
\fs{\sigma_1} \fs{\sigma_2} \es{\sigma}
\exists u, v, w \ldotp
\sigma_1(\pointsto{h}{u} * \acc{o} * \acc{y})
\land
\sigma(\pointsto{h}{u} * \pointsto{o}{v} * \pointsto{y}{w})
\land
\sigma_2(\acc{h} * \pointsto{o}{v} * \pointsto{y}{w})
}}

\newcommand{\functionCut}{\ensuremath{x, y := \mathit{split}(l)}}
\newcommand{\functionComp}{\ensuremath{o := \mathit{compute}(x, h)}}
\newcommand{\functionFinal}{\ensuremath{r := [o] + [y]}}

\newcommand{\fullStatement}{\ensuremath{\functionCut{} \cseq \functionComp{} \cseq \functionFinal{}}}



\subsection{Hyper-Triples: Hyperproperties and Ownership }
\label{subsec:hyper-triples-tour}

HSL's judgments are \emph{hyper-triples} of the form
$\simpleHoare{P}{C}{Q}$,
where $C$ is a program statement,
and $P$ and $Q$  are \emph{hyper-assertions},
\ie predicates over sets of states.
Intuitively, the hyper-triple $\simpleHoare{P}{C}{Q}$ holds
iff for every set of initial states $S$ satisfying $P$,
the set $S'$ of states reachable by executing $C$ in any state from $S$ satisfies $Q$.
HSL judgments can express arbitrary program hyperproperties, that is, properties relating the initial and final states of multiple executions of a program, such as determinism, generalized non-interference, and the existence of a leak of confidential data. Like HHL, HSL obtains this expressiveness by allowing hyper-assertions to contain arbitrary alternations of  universal and existential quantification over states.

However, HSL goes substantially beyond HHL by supporting heap-manipulating programs. In contrast to HHL, states in HSL contain a heap,
and HSL assertions may express properties about the heap using standard SL operators, in particular, the points-to predicate $\mapsto$ and separating conjunction $*$.
For example, the triple
$$\simpleHoare{\fs{\sigma} \sigma(\pointsto{x}{5} * \acc{y})}{\cassign{[y]}{[x]+1}}{\fs{\sigma} \sigma(\pointsto{x}{5} * \pointsto{y}{6})}$$
is equivalent to the standard SL triple $\normalHoare{\pointsto{x}{5} * \acc{y}}{\cassign{[y]}{[x] + 1}}{\pointsto{x}{5} * \pointsto{y}{6}}$.
Both the pre- and the postcondition quantify universally over all states, thereby expressing a property of a single execution of the assignment (we will show examples of proper hyperproperties below). 
The precondition expresses that all initial states $\sigma$ own the heap locations $x$ (containing the value $5$) and $y$. Since HSL assertions generally relate states from multiple executions, assertions explicitly indicate which state they refer to, as in $\sigma(\ldots)$. 
The postcondition expresses that all reachable states own the same locations, and that the value stored at location $y$ is now $6$. 
The standard separating conjunction $*$ enforces that $x$ and $y$ point to disjoint parts of the heap, such that $x$'s value is not affected by the assignment.
As we will see in the next subsection, at the core of HSL lies a hyper separating conjunction $\star$, which generalizes the standard $*$ to express properties of an entire set of heaps rather than a single heap.

The hyper-triple $\simpleHoare{\es{\sigma} \sigma(p)}{C}{\es{\sigma} \sigma(q)}$
corresponds to
the sufficient separation incorrectness logic (SSIL)~\cite{ascariRevealingSourcesMemory2025} triple $\normalHoare{p}{C}{q}$,
which expresses that executing $C$ in any state satisfying $p$ can (non-deterministically) reach a state satisfying $q$.
For example, the hyper-triple
$$\simpleHoare{\es{\sigma} \sigma(r \mapsto 1)}
{\cread{x}{r} \cseq \cassign{y}{\mathit{randInt(1, 6)}} \cseq \cassign{z}{x + y}
}
{\es{\sigma} \sigma(r \mapsto 1 * z = 6)}$$
expresses that executing the statement
in an initial state with heap location $r$ allocated and set to $1$,
can reach a final state satisfying additionally $z = 6$.

A key use case of reachability reasoning as in the previous example is to prove the presence of a bug. To this end, following Incorrectness~\cite{ohearnIncorrectnessLogic2019} and Outcome Logic~\cite{zilbersteinOutcomeLogicUnifying2023},
our assertions distinguish \emph{normal states} $\sigma$, for which we write $\sigma(\ok: \ldots)$,
from \emph{error states}, for which we write $\sigma(\er: \ldots)$.
For example, 
the hyper-triple
$$
\simpleHoare{\es{\sigma} \sigma(\ok: r \mapsto \_ )}
{\cfree{r} \cseq \cfree{r}}
{\es{\sigma} \sigma(\er: \top)}
$$
expresses the reachability of an error state, that is, the presence of a bug.
For this statement, one could also write a stronger triple expressing that all executions \emph{will} fail, by weakening the precondition and using universal quantification in the postcondition.
To avoid clutter, we write $\sigma(p)$ short for $\sigma(\ok:p)$.

\paragraph{Expressing hyperproperties}

Let us now illustrate how HSL can express hyperproperties that relate values stored on the heap.
The following triple expresses that the function $\mathit{split}$, which takes as input one pointer and returns two pointers, is deterministic (\ie two calls with the same value at heap location $l$ will result in the same values at heap locations $x$ and $y$):
\largeHoare{\precut{}}{\functionCut{}}{\postcut{}}
The precondition expresses that all initial states have ownership of the heap location $l$,
and that all states have the same value $u$ stored at location $l$.
The postcondition then expresses that all reachable states have ownership of the two returned pointers $x$ and $y$, with the same values $v$ and $w$, resp.
Thus, these output values depend only on the input stored at $l$.

A key feature of HSL is that universal and existential quantification can be combined to express, for instance, $\forall^* \exists^*$ and $\exists^* \forall^*$ hyperproperties. An important example is generalized non-interference (GNI)
\cite{mcculloughSpecificationsMultiLevelSecurity1987,mcleanGeneralTheoryComposition1996}, a standard notion of secure information flow for non-deterministic programs which, to the best of our knowledge, no other separation logic can express.
GNI requires that for any two executions that agree on the public (low-sensitivity) inputs,
there exists a third execution with the same low inputs that has the same high inputs as the first execution and the same low outputs as the second execution.
The following triple expresses that the function $\mathit{compute}$, which takes two pointers $x$ (low) and $h$ (high)
and returns one pointer $o$, satisfies GNI:
\largeHoare{\preComp{}}{\functionComp{}}{%
\fs{\sigma_1} \fs{\sigma_2} \es{\sigma}
\exists u, v \ldotp
\sigma_1(\acc{o} * \pointsto{h}{v})
\land
\sigma(\pointsto{o}{u} * \pointsto{h}{v})
\land
\sigma_2(\pointsto{o}{u} * \acc{h})
}
In this triple, we assume for simplicity that $\mathit{compute}$ does not modify the value stored in $h$; this property could be expressed using logical variables, as usual~\cite{dardinierHyperHoareLogic2024}.

HSL can express and reason about arbitrary program hyperproperties, such as the examples above, uniformly within one logic. This expressiveness allows one in particular to combine different program hyperproperties within the same proof, a flexibility that other, more specialized logics do not offer. In the rest of this section, we will show that HSL also enables uniform separation-logic reasoning across different program hyperproperties.%

\subsection{Hyper Separating Conjunction}
\label{subsec:hyper-star-tour}

At the core of HSL is the \emph{hyper separating conjunction} $\star$,
which generalizes standard separating conjunction to hyper-assertions.
Hyper separating conjunction allows one to (de)compose the heaps of all states in a set of states uniformly for both ok and error states and for both universally- and existentially-quantified states. Intuitively, a set of states $S$ satisfies the hyper-assertion $P \star Q$
iff there exist two (possibly overlapping) sets $S_P \models P$ and $S_Q \models Q$ such that each state $\sigma \in S$ can be partitioned into two states $\sigma_p \in S_P$ and $\sigma_q \in S_Q$ whose heaps are disjoint.

Hyper separating conjunction lifts separation logic reasoning to sets of states, which enables concise proofs of program hyperproperties, as HSL's rule for memory allocation illustrates:
$$
\newrule{Alloc}{rule:Alloc-Key}
[\baseok[P]]
[x\notin\const{fv}(P)]
{
\hoare{P}
{\calloc{x}}
{P \star (\fs{\sigma} \sigma(\ok: \acc{x}))}
}
$$
This rule expresses that allocation preserves any hyperproperty $P$ and that all reachable states additionally have ownership of the newly-allocated location $x$
(we will discuss in \secref{subsec:formal:rules} why the preservation of $P$ is built into the rule).
The two premises ensure that $P$ does not mention error states (which we will motivate later) or $x$ (such that $P$ is not affected by the assignment).

The hyper separating conjunction in the postcondition elegantly composes the heap of each state constrained by $P$ with the newly-allocated memory. As an example, consider 
the precondition of the function $\mathit{split}$ discussed above, that is, $P \triangleq \precut{}$.
Using laws
(formally proven in our mechanization),
we can distribute the new ownership of $x$ over the quantified states $\sigma_1$ and $\sigma_2$ as follows:
\begin{align*}
&(\precut{}) \star (\fs{\sigma} \sigma(\acc{x})) \\
= \;
&\fs{\sigma_1} \fs{\sigma_2} 
\exists u \ldotp
(\sigma_1(\pointsto{l}{u} * \acc{x})) \land (\sigma_2(\pointsto{l}{u} * \acc{x}))
\end{align*}

On the other hand, if $P$ quantified existentially over states, for instance, $P \triangleq \es{\sigma} \sigma(\pointsto{y}{1})$ then the postcondition would entail
$\es{\sigma}\sigma(\pointsto{y}{1} * \pointsto{x}{\_})$, correctly expressing the reachability of a state owning both $y$ and $x$. These examples illustrate that $\star$ expresses the intended property uniformly for different quantifiers in the conjuncts. We will see in \secref{sec:key-ideas} how it achieves that.

\subsection{Local Reasoning and Framing}
\label{subsec:frame-rule-tour}


A key benefit of separation logic is that properties can be proven \emph{locally} considering only the relevant portion of the heap and then applied in a larger context via the frame rule.
HSL supports local reasoning via the following \emph{generalized frame rule}, which lifts framing to sets of states:
$$
\newrule{Frame}{rule:Frame-Key}
[\hoare{P}{C}{Q}]
[\text{No $\exists \tb{\_}$ in $F$}]
[F \models \fs{\sigma} \sigma(\ok: \top)]
[\fv[F]\cap\md[C]=\emptyset]
{
\hoare{P \star F}{C}{Q \star F}
}
$$
Analogously to the standard frame rule, this rule allows us to extend a valid triple $\hoare{P}{C}{Q}$ with a \emph{frame} $F$, which describes properties unaffected by $C$.
However, in contrast to the standard frame rule, the assertions $P$, $Q$, and $F$ in our generalized frame rule are all hyper-assertions and combined using our hyper separating conjunction, enabling framing in proofs of hyperproperties. 
This rule has three%
\footnote{As we will see in \secref{subsec:formal:syntax},
we need a fourth restriction for assertions containing \emph{scaffold} variables.}
side conditions to ensure soundness.
First, the frame $F$ is not allowed to contain existential quantification over states,
as such existential quantifiers in the postcondition would guarantee reachability and thus termination, 
while $C$ is not guaranteed to terminate in general.%
\footnote{We believe that this restriction could be lifted for terminating programs.}
Second, all states described by the frame must be $\ok$ states; this condition is explained when defining validity (\secref{subsec:formal:hyper-triples}).
The third restriction is standard: it ensures the frame’s free variables ($\fv$) aren’t modified ($\md$) by the command $C$.


\begin{figure}
\tiny
\leqnomode
\begin{align}
	    \hspace{10pt}
		&\hasrt{\fullPre{}} \label{line:fullPre} \\
\models &\hasrt{\left( \precut{} \right) \star \smash{\underbrace{\frameA{}}_{\text{frame 1}}}} \label{line:precut_frame} \\
		&\hasrt{\precut{}} \label{line:precut} \\
		&\functionCut{} \\
	    &\hasrt{\postcut{}} \label{line:postcut} \\
	    &\hasrt{\left( \postcut{} \right) \star \smash{\overbrace{\frameA{}}^{\text{frame 1}}}} \label{line:postcut_frame} \\
\models &\hasrt{\intermediateAssertion{}} \label{line:intermediate} \\
\models &\hasrt{\left( \preComp{} \right) \star \smash{\underbrace{(\frameB{})}_{\text{frame 2}}}} \\
		&\hasrt{\preComp{}} \label{line:precomp} \\
		&\functionComp{} \\
	    &\hasrt{\postComp{}} \label{line:postcomp} \\
	    &\hasrt{\left( \postComp{} \right) \star \smash{\overbrace{(\fs{\sigma_1}\fs{\sigma_2}\en{w}\ldots)}^{\text{frame 2}}}} \label{line:postcomp_frame}\\
\models &\hasrt{\fullPostBefore{}} \label{line:preassign} \\
		&\functionFinal{} \\
	    &\hasrt{\fullPost{}} \label{line:postassign}
\end{align}
\reqnomode
\caption{Local reasoning with hyper-triples. The proof contains two applications of our generalized frame rule, one with a $\forall$-hyperproperty and one with a $\forall\forall$-hyperproperty.}
\label{fig:local-reasoning-example}
\end{figure}

The example in \figref{fig:local-reasoning-example} illustrates our frame rule by proving that the program in black
satisfies generalized non-interference (GNI), a property that is beyond the reach of existing separation logics due to its quantifier alternation in the postcondition. We assume the value at input location $l$ has low-sensitivity (public)
and the value at input location $h$ has high-sensitivity (confidential).
Analogously to $\mathit{compute}$ above,
we verify the following triple:
\largeHoare{\fullPre{}}{\fullStatement}{\fullPost{}}

\emph{Step 1: Framing a $\forall$-property to verify the first call.}
To verify the call to $\mathit{split}$, we need to frame the ownership of the pointer $h$ via the frame rule:
We do so by rewriting the precondition on line \ref{line:fullPre} into the form on line \ref{line:precut_frame},
where we separate the part needed for calling $\mathit{split}$ (line \ref{line:precut}) from the part we frame, $\frameA{}$.
This allows us to use the specification of $\mathit{split}$ (lines \ref{line:precut}--\ref{line:postcut})
to obtain the postcondition on line \ref{line:postcut_frame}, where we add back the framed part $\frameA{}$.
This use of the frame rule is similar to the one in standard SL\@.
By distributing the frame over the conjuncts of the postcondition (line~\ref{line:intermediate}), we get that each execution has ownership of $x$, $y$, and $h$,
and that the values stored at $x$ and $y$ are the same across all executions.

\emph{Step 2: Framing a $\forall \forall$-hyperproperty to verify the second call.}
The call to $\mathit{compute}$ does not need ownership of $y$,
as it accesses only $x$ and $h$.
As before, we could extract ownership to $y$ via the frame $\fsok{\acc{y}}$, but this would lose the fact that all executions have the same value stored at $y$.
To preserve this information, we use the frame $\frameB{}$ instead, which is supported by our generalized frame rule.
Lines~\ref{line:precomp}--\ref{line:postcomp} then use the specification of $\mathit{compute}$ to obtain the postcondition after the call,
and line~\ref{line:postcomp_frame} adds back the frame to the postcondition.
As before, lines~\ref{line:postcomp_frame}--\ref{line:preassign} distribute the frame over the conjuncts of the postcondition. 

\emph{Step 3: Reading from the heap.}
The postcondition follows from line~\ref{line:preassign} and the read operation. In particular, since the values stored in $o$ and $y$ are the same in $\sigma$ and $\sigma_2$, their sum $r$ is also the same.

This example demonstrates that HSL can verify complex hyperproperties, such as GNI with its combination for universal and existential quantification. It combines $\forall$, $\forall \forall$, and $\forall \forall \exists$-hyperproperties within the same proof, demonstrating the logic's versatility. In particular, the generalized frame rule is used to frame both simple ownership information and hyperproperties. In the next section, we survey the key technical innovations that enable this flexibility, before formalizing the logic.

\section{Key Technical Ingredients}
\label{sec:key-ideas}
\setcounter{equation}{0}
\setcounter{parentequation}{0} 

\subsection{Hyper Separating Conjunction}
\label{subsec:key:hyper-star}

Existing separation logics for hyperproperties define separating conjunction pointwise over the elements of a(n ordered) k-tuple. This simple definition does not carry over to HSL, where assertions are expressed over (unordered) sets of states. 

To motivate our definition of hyper separating conjunction, consider the hyper-assertion 
\begin{equation}\label{eqn:forallexists}
\big(\fs{\sigma}\sigma(x\mapsto5)\big)\star\big(\es{\sigma}\sigma(y\mapsto5)\big)
\end{equation}
Intuitively, it should express that
(1)~both $\fs{\sigma}\sigma(x\mapsto5)$ and $\es{\sigma}\sigma(y\mapsto5)$ hold (that is, act as a conjunction), and
(2)~the resources denoted by the conjuncts are disjoint (that is, be separating).
In other words, in any set of states satisfying this hyper-assertion, all states should satisfy $\pointsto{x}{5}$, and at least one state should satisfy $\pointsto{x}{5} * \pointsto{y}{5}$.

This example demonstrates that the hyper separating conjunction needs to preserve both the \emph{upper bounds} on the set of states expressed by each conjunct (via universal quantification, as in our first conjunct) and the \emph{lower bounds} expressed via existential quantification, as in our second conjunct, ensuring the existence of at least one \emph{witness state} satisfying the existential condition. We achieve both with the following definition:
\begin{equation}\label{eqn:hyperstar}
S\models P\star Q\defiff\exists S_P\models P\ldotp\exists S_Q\models Q\ldotp S\subseteq S_P*S_Q\land\pw[S][S_P][S_Q]
\end{equation}

\paragraph{Preserving upper bounds}
The first conjunct of the definition ensures that the upper bounds expressed by $P$ and $Q$ are preserved: For sets of states $S_P$ and $S_Q$ that satisfy $P$ and $Q$, respectively, any set of states $S$ satisfying the hyper separating conjunction contains \emph{at most} the states in the (standard) separating conjunction of $S_P$ and $S_Q$\footnote{Our definition treats the sets of states $S_P$ and $S_Q$ as assertions, which allows us to apply separating conjunction.}. That is, each state $\sigma\in S$ can be decomposed 
into $\hAdd[\sigma][\sigma_P][\sigma_Q]$ such that $\sigma_P\in S_P$ and $\sigma_Q\in S_Q$, where $\hAdd$ denotes the standard state-combining predicate from separation logic: $\hAdd[\sigma][\sigma_P][\sigma_Q]$ holds whenever the three program states agree on the same program store, and the heap component of $\sigma$ results from combining the disjoint heaps of $\sigma_P$ and $\sigma_Q$.

Due to the first conjunct in~(\ref{eqn:hyperstar}), every state in a set of states satisfying the hyper-assertion~(\ref{eqn:forallexists}) includes $\pointsto{x}{5}$. Moreover, every set of states satisfying the hyper-assertion
$\big(\fs{\sigma}\sigma(x\mapsto5)\big)\star\big(\fs{\sigma}\sigma(y\mapsto5)\big)$
includes only states with $x\neq y$. Consequently, 
$\big(\fs{\sigma}\sigma(x\mapsto5)\big)\star\big(\fs{\sigma}\sigma(x\mapsto5)\big)$
is satisfied only by the empty set. These examples demonstrates that hyper separating conjunction implies non-aliasing, thereby preserving one of the key properties of separating conjunction. 

However, the first conjunct of our definition~(\ref{eqn:hyperstar}) does \emph{not} preserve the lower bounds of the conjuncts. For instance, the empty set of states satisfies it for the hyper-assertion~(\ref{eqn:forallexists}) (by choosing the empty set for $S_P$), but does not include the witness state required by the second conjunct of~(\ref{eqn:forallexists}).

\paragraph{Preserving lower bounds}
The purpose of the second conjunct $\pw[S][S_P][S_Q]$ of our definition~(\ref{eqn:hyperstar}) is to preserve the lower bounds expressed by $P$ and $Q$. It consists of two symmetrical conjuncts, which we abbreviate as $\plw[S][S_P][S_Q]$ and $\prw[S][S_P][S_Q]$\footnote{pw, plw, and prw stand for ``preserve witnesses'', ``preserve left witnesses'', and ``preserve right witnesses'', respectively.}.
This conjunct addresses the lower bound, ensuring that the states from $S_P$ and $S_Q$ witnessing the (potential) existential quantifiers from $P$ and $Q$, respectively, are preserved in $S$:
\begin{align*}
	\pw[S][S_P][S_Q]\defiff & (\forall\sigma_P\in S_P\ldotp\exists\sigma_Q\in S_Q\ldotp\exists\sigma\in S\ldotp\hAdd[\sigma][\sigma_P][\sigma_Q])
	\tag{$\plw[S][S_P][S_Q]$} \\
	\land\  & 				  (\forall\sigma_Q\in S_Q\ldotp\exists\sigma_P\in S_P\ldotp\exists\sigma\in S\ldotp\hAdd[\sigma][\sigma_P][\sigma_Q])
	\tag{$\prw[S][S_P][S_Q]$}
\end{align*}

The first preservation conjunct, $\plw$, guarantees that all states $\sigma_P$ (in particular, the witness states) from $S_P$
have a corresponding state $\sigma \in S$;
$\prw$ ensures a symmetric preservation for the states of $S_Q$.
In our example, $\prw$ ensures that any set $S$ satisfying the hyper-assertion~(\ref{eqn:forallexists}) contains a state $\sigma$ satisfying $\pointsto{x}{5} * \pointsto{y}{5}$.

While the preservation predicates are not necessary for the soundness of the logic, they play a pivotal role in enabling reachability reasoning.
In their absence, rules involving $\star$---such as \ruleref{rule:Alloc-Key}---would fail to propagate witness information.
In particular, for any satisfiable $P$, the postcondition, $P\star(\fs{\sigma}\sigma(\ok:x\mapsto\sth))$, would be trivially satisfied by $\emptyset$, eliminating any reachability guarantees.

\begin{figure}
\begin{center}
	\scalebox{0.6}{%
		\begin{tikzpicture}[every node/.style={draw, circle, thick}]
			\begin{scope}[shift={(-6,0)}] 
				\draw[blue] (0,0) ellipse (2.4cm and 1cm);
				\fill[green] (-1.5,0) circle (0.6cm);
				\fill[red] (-0.6,0.5) circle (0.3cm);
				\fill[cyan] (-0.5,-0.5) circle (0.4cm);
				\fill[yellow] (0.5,0.1) circle (0.7cm);
				\fill[gray] (1.2,-0.6) circle (0.2cm);
				\fill[pink] (1.7,0.1) circle (0.4cm);
				\node[draw=none] at (0,1.3) {$S$};
				
				\draw[thick] (0,-2) -- (0,-2.8);
				\draw[thick] (-1,-1.5) -- (0,-2);
				\draw[thick] (0,-2) -- (1,-1.5);
				\node[draw=none] at (0,-1.35) {$S\subseteq S_P*S_Q$};
				\node[draw=none] at (-1.05,-2.35) {$\plw[S][S_P][S_Q]$};
				\node[draw=none] at (1.08,-2.35) {$\prw[S][S_P][S_Q]$};
				
				\coordinate (L) at (-3,-3.5);
				\draw[blue] (L) ellipse (2.4cm and 1cm);
				\fill[green] ($(L)+(-1.5,0)$) -- ++(70:0.6cm) arc (70:320:0.6cm) -- cycle;
				\fill[red] ($(L)+(-0.6,0.5)$) circle (0.3cm);
				\fill[cyan] ($(L)+(-0.5,-0.5)$) -- ++(240:0.4cm) arc (240:360:0.4cm) -- cycle;
				\fill[yellow] ($(L)+(0.5,0.1)$) -- ++(90:0.7cm) arc (90:360:0.7cm) -- cycle;
				\node[draw=none, text=gray] at ($(L)+(1.2,-0.6)$) {$\varnothing$};
				\fill[pink] ($(L)+(1.7,0.1)$) -- ++(140:0.4cm) arc (140:320:0.4cm) -- cycle;
				\node[draw=none] at ($(L)+(0,1.3)$) {$S_P$};
				
				\coordinate (R) at (3,-3.5);
				\draw[blue] (R) ellipse (2.4cm and 1cm);
				\fill[green] ($(R)+(-1.5,0)$) -- ++(-40:0.6cm) arc (-40:70:0.6cm) -- cycle;
				\node[draw=none, text=red] at ($(R)+(-0.6,0.5)$) {$\varnothing$};
				\fill[cyan] ($(R)+(-0.5,-0.5)$) -- ++(0:0.4cm) arc (0:240:0.4cm) -- cycle;
				\fill[yellow] ($(R)+(0.5,0.1)$) -- ++(0:0.7cm) arc (0:90:0.7cm) -- cycle;
				\fill[gray] ($(R)+(1.2,-0.6)$) circle (0.2cm);
				\fill[pink] ($(R)+(1.7,0.1)$) -- ++(-40:0.4cm) arc (-40:140:0.4cm) -- cycle;
				\node[draw=none] at ($(R)+(0,1.3)$) {$S_Q$};
			\end{scope}
			
			\begin{scope}[shift={(6,0)}] 
				\draw[blue] (0,0) ellipse (2.4cm and 1cm);
				\fill[cyan] (-1,0) circle (0.6cm);
				\fill[cyan] (1,0) -- ++(70:0.6cm) arc (70:320:0.6cm) -- cycle;
				\node[draw=none] at (-1,0) {$\sigma$};
				\node[draw=none] at (0.8,-0.1) {$\sigma_1$};
				\node[draw=none] at (0,1.3) {$S$};
				
				\draw[thick] (0,-2) -- (0,-2.8);
				\draw[thick] (-1,-1.5) -- (0,-2);
				\draw[thick] (0,-2) -- (1,-1.5);
				\node[draw=none] at (0,-1.35) {$S\subseteq S_P*S_Q$};
				\node[draw=none] at (-1.05,-2.35) {$\plw[S][S_P][S_Q]$};
				\node[draw=none] at (1.08,-2.35) {$\prw[S][S_P][S_Q]$};

				\coordinate (L) at (-3,-3.5);
				\draw[blue] (L) ellipse (2.4cm and 1cm);
				\fill[cyan] ($(L)+(-1.4,0)$) circle (0.6cm);
				\fill[cyan] (L) -- ++(70:0.6cm) arc (70:320:0.6cm) -- cycle;
				\node[draw=none, text=cyan] at ($(L)+(1.4,0)$) {$\varnothing$};
				\node[draw=none] at ($(L)+(0,1.3)$) {$S_P$};
				\node[draw=none] at ($(L)+(-1.4,0)$) {$\sigma$};
				\node[draw=none] at ($(L)+(-0.2,-0.1)$) {$\sigma_1$};
				
				\coordinate (R) at (3,-3.5);
				\draw[blue] (R) ellipse (2.4cm and 1cm);
				\node[draw=none, text=cyan] at ($(R)+(-1.4,0)$) {$\varnothing$};
				\fill[cyan] (R) -- ++(-40:0.6cm) arc (-40:70:0.6cm) -- cycle;
				\fill[cyan] ($(R)+(1.4,0)$) circle (0.6cm);
				\node[draw=none] at ($(R)+(0,1.3)$) {$S_Q$};
				\node[draw=none] at ($(R)+(1.4,0)$) {$\sigma$};
				\node[draw=none] at ($(R)+(0.3,0.05)$) {$\sigma_2$};
			\end{scope}
			
			\draw[line width=1.5pt, dashed] (0,1.6) -- (0,-4.5);
	\end{tikzpicture}}
\end{center}

\caption{Illustrations of two applications of the hyper separating conjunction. Circular sectors denote the heap component of states; identical (resp.\ different) colors indicate identical (resp.\ different store components). $\varnothing$ denotes a state with empty heap.}
\label{fig:illustration-hyper-star}
\vspace{-0.5cm}
\end{figure}

\paragraph{Examples.}
\figref{fig:illustration-hyper-star} provides further intuition for our definition~(\ref{eqn:hyperstar}) of hyper separating conjunction. The diagram on the left shows a simple instance of the hyper separating conjunction, where each state in $S$ corresponds to exactly one composition of a state in $S_P$ and a state in $S_Q$, thereby satisfying both conjunctions of the definition. The diagram on the right shows more subtle cases. In particular, a state $\sigma\in S$ may be obtained in multiple ways ($\hAdd[\bm\sigma][\sigma][\varnothing]$, $\hAdd[\bm\sigma][\sigma_1][\sigma_2]$ and $\hAdd[\bm\sigma][\varnothing][\sigma]$). Moreover,
the states in $S_P$ (and $S_Q$) may participate in multiple combinations ($\hAdd[\sigma][\bm\sigma_{\bm1}][\sigma_2]$ and $\hAdd[\sigma_1][\bm\sigma_{\bm1}][\varnothing]$).
Note that in the right diagram, multiple sets $S$ satisfy definition~(\ref{eqn:hyperstar}) for the given $S_P$ and $S_Q$; for example, extending the depicted $S$ by $\varnothing$ still satisfies all requirements.
In contrast, in the left diagram, $S$ is uniquely determined by the given $S_P$ and $S_Q$. In particular, adding $\varnothing$ to $S$ would violate the first conjunct of definition~(\ref{eqn:hyperstar}) because the two $\varnothing$ states have different stores and therefore cannot be combined.

So far, we have illustrated hyper separating conjunction on conjuncts with a single quantifier, but our definition supports arbitrary quantifier alternations. For example, each set of states $S_P$ satisfying the first conjunct of the hyper-assertion
\[
\big(\fs{\sigma}\es{\sigma'}\en{n}\sigma(x=n)\land\sigma'(x=n+1)\big)\star\big(\fs{\sigma}(x\mapsto\sth)\big)
\]
is either empty or contains infinitely many states (with increasing values of $x$). Definition~(\ref{eqn:hyperstar}) retains these infinitely many witness states (via $\plw$) in $S$ and ensures (via its first conjunct) that each state in $S$ owns the heap location at that particular address.

\paragraph{Labeled States}
Thus far, we have presented the (hyper) separating conjunction for unlabeled states. In the following, we adapt the definition to HSL's labeled states (see \secref{subsec:hyper-triples-tour}). Existing solutions do not apply to our setting. \citet{raadLocalReasoningPresence2020} treat the label merely as a tag in the postcondition while applying the separating conjunction to unlabeled states; this approach does not apply to our hyper-assertions, which constrain sets of states with possibly different labels. \citet{zilbersteinOutcomeSeparationLogic2024} use an asymmetric definition where only one argument carries labels; in contrast, we aim for our hyper separating conjunction to be symmetric, which is especially important for future extensions to concurrency.

We define both standard and hyper separating conjunctions by
lifting the underlying state-combining predicate $\hAdd$ from unlabeled to labeled states.
$\hMonadAdd[(\sigma)\el][(\sigma_1)_k][(\sigma_2)_l]$ holds iff 
(1)~$\hAdd[\sigma][\sigma_1][\sigma_2]$ and
(2)~$\epsilon=\ok$ iff $k=\ok$ and $l=\ok$.
The case that all three labels are $\ok$ corresponds to the standard, unlabeled program states.
As a direct consequence, we obtain that $(\ok:p)*(\ok:q)=\ok:p*q$, where we overload $*$ to denote both labeled and unlabeled separating conjunctions.

To understand the intuition for the case $k=\er,l=\ok$, consider an execution that starts in an ok-state and results in an error state $(\sigma_1)_k$, for instance, by causing a runtime error. Extending both the initial and the final state with the ok-state $(\sigma_2)_l$ will not eliminate the runtime error, so the extended final state is still an error state\footnote{As we will discuss in \secref{subsec:key:uk-states}, some errors (specifically, \emph{domain violation errors}) do not behave this way.
Here, we discuss runtime errors such as assertion violations, null dereferences, or division by zero, which would also occur in larger heaps.}. The case $k=\ok,l=\er$ is symmetric. As a direct consequence, we obtain that $(\er:p)*(\ok:q)=(\ok:p)*(\er:q)=\er:p*q$.
A similar argument justifies that combining two error states yields an error state and hence $(\er:p)*(\er:q)=\er:p*q$.

\subsection{One for All: Enforcing $\forall$-Frames}
\label{subsec:key:embedding-fr}

Let us now define the validity of hyper-triples.
As discussed in \secref{sec:tour}, intuitively the triple $\simpleHoare{P}{C}{Q}$ is valid iff for every set of states $S$ satisfying the precondition $P$, executing $C$ from each state in $S$ yields a set of states satisfying the postcondition $Q$.
We denote this \emph{candidate} validity $\models_0 \simpleHoare{P}{C}{Q}$.

\paragraph{Sound framing}
Unfortunately, this candidate definition does not support the frame rule, as shown by the following \emph{invalid} example:%

{
\footnotesize
$$
\applyrule{rule:Frame-Key}
[\models_0 \simpleHoare{\es{\sigma}\sigma(\emp)}{\Alloc{x}}{\es{\sigma}\sigma(\acc{x} * x=42)}]
{\not{\models_0} \simpleHoare{
	\underbrace{(\es{\sigma}\sigma(\emp)) \star (\fs{\sigma} \sigma(\acc{42}))}_{%
	\supseteq(\es{\sigma}\sigma(\acc{42})) \land (\fs{\sigma} \sigma(\acc{42}))}
}{\Alloc{x}}%
{\underbrace{(\es{\sigma}\sigma(\acc{x} * x=42)) \star (\fs{\sigma} \sigma(\acc{42}))}_{\bot}}
}
$$
}

According to our candidate definition, the first triple is valid:
An execution starting with an empty heap might (non-deterministically) allocate address $42$, yielding a state satisfying the postcondition.
However, this specific address is ruled out when combined with the frame $\fs{\sigma}\sigma(\acc{42})$ in the second triple, as address $42$ is already allocated. Consequently, framing yields an invalid triple with a satisfiable precondition (\eg a singleton set with a state $\sigma$ satisfying $\acc{42}$), but an unsatisfiable postcondition.

This problem with the candidate definition is caused by the incompatibility between the non-local nature of allocation and the local reasoning embodied by the frame rule.
Existing work on unary reasoning (\ie not involving hyperproperties) has addressed this problem in two different ways.
\citet{ascariSufficientIncorrectnessLogic2024} fix a set of \emph{syntactic} rules and then prove that this set of rules guarantees the soundness of the proof system.
However, this syntactic approach is not easily extensible, as adding a new rule that is sound by itself (\wrt\ the semantic definition of validity) is not guaranteed to preserve the soundness of the overall system.
\citet{zilbersteinOutcomeSeparationLogic2024} instead instrument the semantics of the programming language with a notion of \emph{allocator} (a function that returns a set of heap locations that can be allocated in a given state), and then define a triple to be valid iff it holds for all allocators.
While this approach solves the particular problem with allocation in their setting, it requires instrumenting the semantics of the programming language (with the allocator), and does not provide a general solution for other constructs that cause similar issues with framing, as we explain next.

\paragraph{Non-local errors.}
The problem described above is neither specific to allocation nor to reachability reasoning,
but also occurs when operations lead to \emph{non-local} errors,
that is, errors due to insufficient resources that may not occur in a larger heap.
After adding the missing resources with the frame rule, these operations succeed, leading to an incorrect labeling of the resulting states:

{
\footnotesize
$$
\applyrule{rule:Frame-Key}
[\models_0 \simpleHoare{\fs{\sigma}\sigma(\ok: \emp)}{\Free{x}}{\fs{\sigma}\sigma(\er: \emp)}]
{\not{\models_0} \simpleHoare{
	\underbrace{
		(\fs{\sigma}\sigma(\ok: \emp)) \star (\fs{\sigma} \sigma(\ok: \acc{x}))
	}_{
		\fs{\sigma} \sigma(\ok: \acc{x})
	}
}{\Free{x}}%
{\underbrace{
		(\fs{\sigma}\sigma(\er: \emp)) \star (\fs{\sigma} \sigma(\ok: \acc{x}))
	}_{
		\fs{\sigma}\sigma(\er: \acc{x})
	}}
}
$$
}
\noindent
We call such errors, which arise from accessing memory locations that are not currently owned, \emph{domain violation errors}.
Besides for deallocation,
they may occur for heap read and write operations.
In contrast, all other errors (e.g., assertion violations and dereferencing null pointers) are \emph{local} (\ie if they occur in a state, they also occur in any larger state), and thus are compatible with the frame rule.

\paragraph{Baking in the frame rule}
To solve these problems, we take inspiration from existing unary separation logics (\eg \cite{birkedalSimpleModelSeparation2008, chargueraudSeparationLogicSequential2020}) and adapt it to hyperproperty reasoning: We require the validity of framing \emph{as part of} the definition of triple validity. Intuitively, we define the hyper-triple $\simpleHoare{P}{C}{Q}$ to be valid iff $\models_0 \simpleHoare{P \star F}{C}{Q \star F}$ holds for every admissible%
\footnote{That is, a frame $F$ that satisfies the premises of the rule \ruleref{rule:Frame-Key}, see \secref{subsec:frame-rule-tour}.} frame $F$.

This approach yields a simple semantic definition of triple validity that ensures the soundness of the frame rule, even in the presence of reachability reasoning and error states.
In particular, it rules out the incorrect derivations above, as both examples start from an \emph{invalid} triple. In the allocation example, the first triple does not preserve the frame $\fs{\sigma}\sigma(\ok:\acc{42})$; in the de-allocation example, the first triple does not preserve the frame $\fs{\sigma}\sigma(\ok:\acc{x})$.

\paragraph{Simplifying the meta theory}
The above definition of validity is sufficient to obtain a sound logic. It makes the soundness proof of the frame rule straightforward, but requires one to prove that every other rule yields a valid triple for arbitrarily complex frames $F$. 

We make a key observation to simplify these proofs: It is sufficient for validity to require the preservation of \emph{unary} frames, that is, frames with a single $\forall$-quantifier. Thus, intuitively, a hyper-triple $\simpleHoare{P}{C}{Q}$ is valid, denoted $\hoare{P}{C}{Q}$, iff
$\models_0 \simpleHoare{P \star (\fs{\sigma} \sigma(\ok: f))}{C}{Q \star (\fs{\sigma} \sigma(\ok: f))}$ holds for all (admissible) unary frames $f$. 

This weaker definition is sufficient to guarantee the preservation of all $\forall^+$-frames (and, thus, soundness of the frame rule) for the following reasons:

\begin{enumerate}
	\item Any $\forall^+$-hyperassertion $F$ can be decomposed into a (potentially infinite) disjunction of unary $\forall$-assertions
	$\bigvee_i(\fs{\sigma} \sigma(f_i))$.
	For example, $\fs{\sigma_1} \fs{\sigma_2} \en{n} \sigma_1(x=n) \land \sigma_2(x=n)$ is equivalent to $\bigvee_i(\fs{\sigma} \sigma(x=i))$.\footnote{Intuitively, this means that any safety hyperproperty can be seen as a (potentially infinite) disjunction of unary safety properties. Reducing safety hyperproperties to unary properties is known from product constructions; the difference here is that we do not enlarge the states, but instead reflect all possible combinations in a (potentially infinite) disjunction.}

	\item The hyper separating conjunction distributes over (infinite) disjunctions: $P \star \bigvee_i Q_i \equiv \bigvee_i(P \star Q_i)$.

	\item HSL's triples support the (infinite) disjunction rule. That is, if $\hoare{P_i}{C}{Q_i}$ holds for all $i$
	then $\hoare{\bigvee_i P_i}{C}{\bigvee_i Q_i}$ also holds.
\end{enumerate}

We formalize this argument as part of our soundness proof for the frame rule
(see \appref{sec:appendix-soundness}).
The soundness proof of \emph{every other rule} becomes significantly simpler, as they need to consider only unary $\forall$-frames.

\subsection{Local Reachability Reasoning in the Presence of Errors}
\label{subsec:key:uk-states}

Requiring frame preservation as part of triple validity has possibly surprising consequences for reasoning about errors.
As explained in the previous subsection,
our definition of triple validity \emph{correctly} rejects triples where error states originate only from non-local (domain violation) errors,
such as the following triple:
\begin{equation}\label{eqn:non-local-error1}
\simpleHoare{\es{\sigma}\sigma(\ok:\emp)}{\Free{x}}{\es{\sigma}\sigma(\er:\emp)}
\end{equation}
which expresses that de-allocating $x$ may lead to an error when executed in an empty heap.
This triple is \emph{not} valid because there are frames, such as $\fs{\sigma} \sigma(\ok: \acc{x})$, for which the resulting postcondition does \emph{not} hold: adding ownership of the location at $x$ causes the de-allocation to succeed, such that it can no longer reach an error state.
%
Non-local errors are merely an artifact of local reasoning in separation logic, but do not necessarily correspond to actual errors at runtime.

However, our definition of validity also rejects \emph{reasonable} triples, in fact, \emph{any} triple that expresses reachability properties for programs that may cause non-local errors.
For example, the triple 
\begin{equation}
\label{eqn:non-local-error2}
\simpleHoare{\es{\sigma}\sigma(\ok:\acc{x})}{\Free{x}}{\es{\sigma}\sigma(\ok:\pointsto{x}{\bot})} \end{equation}
expresses a reachability property, which clearly holds: starting from a state where $x$ is allocated, freeing $x$ succeeds and may lead to a state where $x$ is no longer owned (in HSL, de-allocated memory contains a special value $\bot$).
However, and perhaps surprisingly, this triple is also \emph{invalid} because it does not preserve the frame $\fs{\sigma} \sigma(\ok:\acc{y})$. The reason is subtle: The precondition (including the frame) guarantees that $\acc{x} * \acc{y}$ (and thus $x \neq y$) holds in at least one state%
; all other initial states are guaranteed only to satisfy $\acc{y}$. In particular,
$x$ and $y$ might alias in some other initial state,
such that freeing $x$ loses the ownership of $y$, which violates the framed conjunct of the postcondition.

Even if not apparent, this undesired side effect of our definition of triple validity is also caused by non-local errors.
To understand why, consider a set with two initial states $\sigma_w$ and $\sigma_u$ satisfying the precondition $\es{\sigma}\sigma(\ok:\acc{x})$:
$\sigma_w$ witnesses the existential quantifier, and thus owns $x$, whereas $\sigma_u$ does not witness the existential quantifier, and thus is unconstrained.
In particular, assume that $x$ and $y$ alias in
$\sigma_u$ (\ie $x=y$), and that $\sigma_u$ does not own $x$.
The executions of $\Free{x}$ from these two states are illustrated in \figref{fig:non-local-errors}:
The top left of \figref{subfig:nle:no-problem} shows the successful
 execution from $\sigma_w$ to $\sigma_w'$ (where ownership of $x$ is removed),
 while the top left of \figref{subfig:nle:problem} shows the erroneous execution from $\sigma_u$ (since $\sigma_u$ does not own $x$).
In the latter case, we record the state in which the (domain violation) error occurred, $\sigma_u$,
with a special label (which will be explained below).

\begin{figure}
	\newcommand{\dx}{5}   
	\newcommand{\dy}{1.1} 
	
	\tikzset{
		commdiagram/.style = {>=latex, x=\dx cm, y=\dy cm},
		state/.style = {},
		harr/.style  = {->},
		varr/.style  = {dashed, ->},
		ok/.style = {text=blue},
		er/.style  = {text=red},
		uk/.style  = {text=orange}
	}
	
	\begin{subfigure}{\textwidth}
		\begin{tikzpicture}[commdiagram]
			\node[state, ok] (A)  at (0,1) {$\ok: \sigma_w$};
			\node[state, ok] (B)  at (1,1) {$\ok: \sigma_w'$};
			\node[state, ok] (C)  at (2,1) {$\ok: \sigma_w''$};
			
			\node[state, ok] (A2) at (0,0) {$\ok: \sigma_w \oplus \sigma_f$};
			\node[state, ok] (B2) at (1,0) {$\ok: \sigma_w' \oplus \sigma_f$};
			\node[state, ok] (C2) at (2,0) {$\ok: \sigma_w'' \oplus \sigma_f$};
			
			\draw[harr] (A)  -- node[above]{$\Free{x}$} (B);
			\draw[harr] (B)  -- node[above]{$C$}        (C);
			\draw[harr] (A2) -- node[above]{$\Free{x}$} (B2);
			\draw[harr] (B2) -- node[above]{$C$}        (C2);
			
			\draw[varr] (A) -- node[right]{$\oplus \sigma_f$} (A2);
			\draw[varr] (B) -- node[right]{$\oplus \sigma_f$} (B2);
			\draw[varr] (C) -- node[right]{$\oplus \sigma_f$} (C2);
		\end{tikzpicture}
		\caption{Successful execution of de-allocation (top left) and a subsequent statement $C$ (top right). Both preserve frames,
				as visualized by the downward arrows.}
		\label{subfig:nle:no-problem}
	\end{subfigure}

	\begin{subfigure}{\textwidth}
		\begin{tikzpicture}[commdiagram]
			
			\node[state, ok] (A)  at (0,1) {$\ok: \sigma_u$};
			\node[state, er] (B)  at (1,1) {$\_: \sigma_u$};
			\node[state, er] (C)  at (2,1) {$\_: \sigma_u$};
			
			\node[state, ok] (A2) at (0,0) {$\ok: \sigma_u \oplus \sigma_f$};
			\node[state, ok] (B2) at (1,0) {$\ok: \sigma_{uf}'$};
			\node[state, ok] (C2) at (2,0) {$\ok: \sigma_{uf}''$ };
			
			\draw[harr] (A)  -- node[above]{$\Free{x}$} (B);
			\draw[harr] (B)  -- node[above]{$C$}        (C);
			\draw[harr] (A2) -- node[above]{$\Free{x}$} (B2);
			\draw[harr] (B2) -- node[above]{$C$}        (C2);
			
			\draw[varr] (A) -- node[right]{$\oplus \sigma_f$} (A2);
			\draw[varr] (B) -- node[right]{?} (B2);
			\draw[varr] (C) -- node[right]{?} (C2);
		\end{tikzpicture}
		\caption{An execution of de-allocation that produces a non-local error (top left). Consequently, the subsequent statement $C$ will \emph{not} be executed and, thus, does not change the state $\sigma_u$ (top right). We omit the label on the resulting states (in red) as they will be discussed later. The execution at the bottom shows the case where de-allocation (and the subsequent execution of $C$) succeeds due to a suitable frame. Neither execution preserves frames (indicated by the question marks).}
		\label{subfig:nle:problem}
	\end{subfigure}

	\begin{subfigure}{\textwidth}
		\begin{tikzpicture}[commdiagram]
			
			\node[state, ok] (A)  at (0,2) {$\ok: \sigma_u$};
			\node[state, uk] (B)  at (1,2) {$\uk: \sigma_u$};
			\node[state, uk] (C)  at (2,2) {$\uk: \sigma_u$};
			
			\node[state, uk] (B1)  at (1,1) {$\uk: \sigma_u \oplus \sigma_f$};
			\node[state, uk] (C1)  at (2,1) {$\uk: \sigma_u \oplus \sigma_f$};
			
			\node[state, ok] (A2) at (0,0) {$\ok: \sigma_u \oplus \sigma_f$};
			\node[state, ok] (B2) at (1,0) {$\ok: \sigma_{uf}'$};
			\node[state, ok] (C2) at (2,0) {$\ok: \sigma_{uf}''$ };
			
			\draw[harr] (A)  -- node[above]{$\Free{x}$} (B);
			\draw[harr] (B)  -- node[above]{$C$}        (C);
			\draw[harr] (A2) -- node[above]{$\Free{x}$} (B2);
			\draw[harr] (B2) -- node[above]{$C$}        (C2);
			\draw[harr] (B1) -- node[above]{$C$}        (C1);
			
			\draw[varr] (A) -- node[right]{$\oplus \sigma_f$} (A2);
			
			\draw[varr] (B) -- node[right]{$\oplus \sigma_f$} (B1);
			\draw[varr] (C) -- node[right]{$\oplus \sigma_f$} (C1);
			
			\draw[varr] (B1) -- node[right]{overapproximates} (B2);
			\draw[varr] (C1) -- node[right]{overapproximates} (C2);
		\end{tikzpicture}
		\caption{The scenario from \figref{subfig:nle:problem}, now correctly labeled as unknown states (in orange). The uk-state $\sigma_u$ in which the non-local error occurred, extended by the frame, overapproximates all states an execution might reach.}
		\label{subfig:nle:solution}
		\vspace{-0.2cm}
	\end{subfigure}
	
	\caption{Illustration of the problem of non-local errors.
		$\sigma_w$ denotes a state where $x$ is allocated, and $\sigma_u$ a state where $x$ is unallocated, but in which $x = y$.
		$\sigma_f$ corresponds to a state with ownership of $y$, that is compatible with $\sigma_w$ in the first diagram, and with $\sigma_u$ in the second and third diagram ($\sigma_f$ is different in each case).}
	\label{fig:non-local-errors}
	\vspace{-0.4cm}
\end{figure}

We now add the frame $\fs{\sigma} \sigma(\ok:\acc{y})$,
\ie we add states $\sigma_f$ with ownership of $y$ to both $\sigma_w$ and $\sigma_u$.%
\footnote{Note that this state $\sigma_f$ is different between $\sigma_w$ and $\sigma_u$.}
The bottom of \figref{subfig:nle:no-problem} shows that the frame $\sigma_f$ is preserved
when executed from the combined state $\sigma_w \oplus \sigma_f$, as it reaches the state $\sigma_w' \oplus \sigma_f$.
In contrast, the bottom of \figref{subfig:nle:problem} shows that the frame $\sigma_f$ is \emph{not} preserved
when executed from the combined state $\sigma_u \oplus \sigma_f$:
Deallocation now succeeds, and thus ownership of $y$ is removed,
as $\sigma_f$ provides ownership of $x$ (since $x$ and $y$ alias in $\sigma_u$).

In general, any execution that results in a non-local error might not preserve the frame because (1) resources in the frame have been consumed (like in our example), (2) heap locations have been modified for which the frame provides ownership, or (3) generally because the program executes subsequent statements rather than staying in a failed state (as illustrated by the statement $C$ in \figref{subfig:nle:no-problem} and \figref{subfig:nle:problem}).

To account for the fact that frames cannot be expected to be preserved once a non-local error has occurred, we adapt the definitions introduced so far in two ways. First, we allow only those states to be labeled as error states in which a \emph{local} error occurred (we will explain below how we label states with non-local errors). As a result, universally-quantified postconditions with ok- or er-labels exclude the occurrence of non-local errors and, thus, are compatible with framing.
Second, we adjust triple validity to \emph{relax} the requirement on \emph{existentially-quantified postconditions}%
\footnote{Technically, the requirement is relaxed for all postconditions that permit states resulting from non-local errors.
In practice, the only such useful postconditions are exactly the existentially-quantified ones.}
(which do not exclude the possibility of a non-local error),
by \emph{effectively} combining the universally-quantified frame only with
quantified ($\ok$ or $\er$) states,
\ie with all states for universal quantifiers,
but only with the witness states for existential quantifiers.
If the existential uses an ok- or er-label, those states cannot be affected by non-local errors and, thus, are also compatible with framing.

The first adjustment ensures that triple \eqref{eqn:non-local-error1} is correctly deemed invalid. The second adjustment interprets the framed postcondition 
$\es{\sigma}\sigma(\ok:\pointsto{x}{\bot}) \star \fs{\sigma} \sigma(\ok:\acc{y})$ of triple~(\ref{eqn:non-local-error2}) as
$\es{\sigma}\sigma(\ok:\pointsto{x}{\bot} *\acc{y})$, such that the triple is valid, as it should.

\paragraph{Unknown states.}

Technically, we handle non-local errors by introducing a third label $\uk$ for \emph{unknown} states, in addition to the existing $\ok$ and $\er$ labels. Executing an operation that causes a non-local error results in an unknown state (whereas local errors result in an error state).
This immediately provides the first adaptation discussed above.

Intuitively, when a statement leads to a non-local error, the resulting $\uk$-state records the state in which the error occurred.
For instance,
$\_: \sigma_u$ in \figref{subfig:nle:problem} is actually labeled as $\uk: \sigma_u$,
as shown in \figref{subfig:nle:solution},
both after executing $\Free{x}$ and after executing $C$ \emph{for any command $C$}.
After
adding the frame $\sigma_f$ to $\sigma_u$,
the de-allocation might or might not lead to a non-local error, depending on whether the frame provides ownership of $x$.
To reflect this uncertainty, we extend the definition of
our state-combining predicate $\hAdd$
(and thus of our hyper separating conjunction)
to yield a $\uk$-label when combining a $\uk$-label (the postcondition) with an $\ok$-label (the frame).
For example, $(\uk: \sigma_u) \hAdd (\ok: \sigma_f) = \uk: (\sigma_u \oplus \sigma_f)$, 
as shown in \figref{subfig:nle:solution}.

Moreover,
executing $\Free{x}; C$ from the larger state $\sigma_u \oplus \sigma_f$
might lead to
a completely different state (for instance, because $C$ modifies variables and allocates new memory), to multiple final states (when $C$ is non-deterministic), or not to any final state (when $C$ does not terminate). This shows that an unknown state $\uk: \sigma$ is an abstraction that may represent zero, one, or multiple concrete states, which may be completely different from $\sigma$.%
\footnote{We can only guarantee that these concrete states agree with $\sigma$ on the variables that do not syntactically appear in $C$.}

Our second adaption above (the relaxed requirement on postconditions)
actually corresponds to overapproximating the reachable possible states represented by unknown states, and thus is not overly conservative%
\footnote{Conservative in the sense that it forgets all information (except the values of logical variables) about $\uk$ states. Not overly conservative in the sense that no information about $\ok$ and $\er$ states is lost.}.
Our triples provide precise information for the standard cases of ok- and er-states; and treat $\uk$-states as overapproximating multiple different concrete states.
This allows HSL to prove strong hyperproperties involving successful and erroneous executions (for errors such as accesses to freed memory, null-pointer dereferences, or failed assertions).
Our adequacy theorem (\thmref{th:adequacy}) captures this formally: if all state-quantifiers in a hyper-assertion $Q$ are labeled with $\ok$ or $\er$, and $Q$ does not contain $\star$, then the relaxation of $Q$ is $Q$.

Even though uk-states are a coarse abstraction, recording the state in which a non-local error occurred (instead of simply recording \emph{that} a non-local error occurred as with \citet{zilbersteinOutcomeSeparationLogic2024}'s \emph{undef} state) is still
important:
When an unknown state $\uk: \sigma_u$ is combined with an \emph{incompatible} frame $\sigma_f$ (for instance because both provide ownership of the same location),
then it is sound to remove this unknown state from the resulting set of states,
as this execution could not have occurred (because the state $\sigma_u \oplus \sigma_f$ before the error is actually inconsistent).
In particular, this allows us to define state composition $\hAdd$ in a way that is associative, yielding an associative hyper separating conjunction $\star$, and thus improved reasoning principles.

\section{Hyper Separation Logic, Formally}
\label{sec:Hyper-Separation-Logic}
In this section, we formally define Hyper Separation Logic (HSL),
including its underlying semantic model and proof rules.
All formal definitions and theorems presented in this section have been mechanized and proven \citep{HyperSeparationLogic_artifact} in Isabelle/HOL~\cite{nipkowIsabelleHOLProof2002}.

\subsection{Programming Language and Semantics}
\label{subsec:formal:semantics}
We define Hyper Separation Logic for a simple language with the following commands:

\[
\begin{aligned}
	C \triangleq &\ \Skip\mid\Assign{x}{e}\mid\Havoc{x}\mid\Assume{b}\mid\Assert{b}\\
	\mid&\ \Alloc{x}\mid\Write{x}{e}\mid\Read{y}{x}\mid\Free{x}\mid(C;C)\mid(C+C)\mid C^*, 
\end{aligned}
\]
where $x$ and $y$ are non-negative%
\footnote{Restricted to non-negative values for simplicity, avoiding a separate type distinction between integers and heap locations. This is not a fundamental limitation, as integers and their operations can be encoded using non-negative ones.}
integer program variables,
$e$ ranges over the usual non-negative integer expressions, and $b$ is a boolean expression. 

\Skip, assignment $\Assign{x}{e}$, and sequential composition are standard.
The command $\Havoc{x}$ assigns a non-deterministic unsigned integer value to $x$.
The $\Assume{b}$ command discards the execution if the current program state does not satisfy $b$, otherwise it acts like $\resv{skip}$.
The $\Assert{b}$ command aborts the execution when $b$ is not satisfied, and otherwise behaves like $\resv{skip}$.
The command $\Alloc{x}$ allocates a heap location and stores it in $x$, $\Write{x}{e}$ updates a heap location,  $\Read{y}{x}$ reads from it, and
$\Free{x}$ deallocates the memory pointed to by $x$.
The $+$ command is nondeterministic choice and $^*$ is nondeterministic iteration.
Using these commands we can define the standard control structures and constrained non-determinism as follows:
\begin{align*}
	\If{b}{C_1}{C_2} &\triangleq(\Assume{b};C_1)+(\Assume{\neg b};C_2) \\
	\While{b}{C} &\triangleq(\Assume{b};C)^*;\Assume{\neg b}\\
	\cHavoc{x}{l}{u} &\triangleq\Havoc{x};\Assume{l\le x\le u}
\end{align*}

A \emph{program state} 
$\sigma=\la s,h\ra$ is a pair consisting of a program store $s$ (a total function from program variables in $\PVars$ to $\mathbb{N}$)
and a program heap $h$ (a finite partial function from heap locations in $\mathbb{N}$ to $\mathbb{N}\cup\{\bot\}$).
Heap location $0$ represents a null-pointer, and value $\bot\notin\mathbb{N}$ indicates that a memory location has been freed:
Following ISL~\cite{raadConcurrentIncorrectnessSeparation2022} and OSL~\cite{zilbersteinOutcomeSeparationLogic2024}, we distinguish between un-allocated ($x \not\in \mathrm{Dom}(h)$) and freed ($h(x)=\bot$) memory,
and we also rule out the re-allocation of freed memory\footnote{This design provides a practical balance, enabling address arithmetic, while keeping the \ruleref{ax:alloc} rule simple and tractable.}.
Arithmetic expressions $e$ are total functions from program stores to $\mathbb{N}$, and boolean expressions $b$ are total functions from program stores to booleans.

A \emph{program configuration} is defined as $c\Coloneqq\la C,\sigma\ra\mid\cDv\mid\cEr$, where $C$ is a command, $\sigma$ is a program state, and $\cDv$ and $\cEr$ are error configurations.
The \emph{small-step semantics} with judgment $c\to c'$ is defined in \appref{sec:appendix}.
The reflexive and transitive closures of $\to$ is denoted by $\to^*$.
Our semantics models failing executions as transitions that end in an error configuration.
As motivated in \secref{subsec:key:uk-states},
we distinguish two kinds of failures: (1)~accesses to un-allocated memory, which we call \emph{domain violations}, and (2)~all other errors, which include accesses to freed memory, null-pointer dereferences, and failed assertions.

\subsection{States, Hyper-Assertions, and Hyper Separating Conjunction}
\label{subsec:formal:state-model}

The hyper separating conjunction, which we presented in \secref{subsec:hyper-star-tour}, works at the level of \emph{hyper-assertions}, \ie sets of extended states.
\emph{Extended states} are program states with two additional pieces of information.
First, we mark the states with $\ok$, $\er$, or $\uk$ labels to indicate whether the corresponding execution is normal, erroneous, or unknown, respectively, as explained in \secref{sec:key-ideas}.
Second, we equip extended states with a \emph{logical store} to record the values of \emph{logical variables}, which remain fixed across executions.
Beyond their usual role of recording the initial values of program variables, logical variables play a crucial role in
relational reasoning (\eg to distinguish different executions, as explained by \citet{dardinierHyperHoareLogic2024}),
and, as we will see in \secref{subsec:formal:uk-res}, they are treated differently from program variables by the relaxation of postconditions.

We thus define an \emph{(extended) state} $\omega$ as a triple $\la\Lambda,\sigma\el\ra$ consisting of a
\emph{logical store} $\Lambda$ (a total function from logical variables in $\const{LVars}$ to $\mathbb{N}$),
an execution label $\epsilon\in\{\ok,\er,\uk\}$, and a \emph{program state} $\sigma$.
\emph{Assertions} are sets of states, and \emph{hyper-assertions} are sets of assertions.

To define the meaning of the (hyper) separating conjunction, we first 
define the \emph{state-combining predicate} $\oplus$,
building on the label combination principles motivated in \secref{sec:key-ideas},
as follows:
$$
\addStates%
[\la\Lambda, (s, h) \el\ra]
[\la\Lambda_1,(s_1, h_1)_{\epsilon_1}\ra]
[\la\Lambda_2,(s_2, h_2)_{\epsilon_2}\ra]
\!\!\defiff\!\!
s = s_1 = s_2 \land 
h = h_1 \uplus h_2 \land
\epsilon=\epsilon_1 \bowtie \epsilon_2\land\Lambda=\Lambda_1 = \Lambda_2
$$
where the \emph{label combination operator} $\bowtie$ is defined
as follows:
$\epsilon_1 \bowtie \epsilon_2$ is $\uk$ if either $\epsilon_1$ or $\epsilon_2$ is $\uk$
(propagating the lack of knowledge about the execution status);
otherwise it is $\er$ if either $\epsilon_1$ or $\epsilon_2$ is $\er$ (as $\ok$ frames can be added to erroneous executions without changing their erroneous status);
and it is $\ok$ if both $\epsilon_1$ and $\epsilon_2$ are $\ok$
(corresponding to the standard case).%
\footnote{While HSL allows framing only with ok-states, the four combinations not involving ok-states are relevant for future extensions, in particular, to concurrency.}
The \emph{(standard) separating conjunction}, $*$, is defined in the usual way:
$p*q\triangleq\{\omega\mid\omega_p\in p\land\omega_q\in q\land\addStates[\omega][\omega_p][\omega_q]\}$.
The usual properties of $*$ such as unit element $\{\la\Lambda,\la s,\emptyset\ra\okl\ra\mid\mathrm{true}\}$, commutativity, associativity and distributivity over (finite and infinite) disjunction are preserved by this extension.

As explained in \secref{subsec:key:hyper-star}, we define 
the hyper separating conjunction $\star$ as follows:
\begin{align*}
&\plw[S][S_1][S_2] \defiff
\forall\omega_1\in S_1\ldotp\exists\omega\in S\ldotp\exists\omega_2\in S_2\ldotp\addStates[\omega][\omega_1][\omega_2]
\defiff
\prw[S][S_2][S_1]
\\
&P\star Q\triangleq\{S\mid S\subseteq S_P*S_Q\land S_P\in P\land S_Q\in Q\land\plw[S][S_P][S_Q]\land\prw[S][S_P][S_Q]\}
\end{align*}
where the predicates $\plw$ and $\prw$
ensure that witnesses are preserved when combining states.

The hyper separating conjunction shares many of the key properties of the standard separating conjunction, including the existence of a unit element%
\footnote{Interestingly, the unit element in this setting is 
$\mathcal{P}(\{\la\Lambda,\la s,\emptyset\ra\okl\ra\mid\mathrm{true}\})$, i.e., the power set of the unit element of the standard separating conjunction.
As we will see, having a unit element is crucial to prove our adequacy theorem,
which requires eliminating the frame in the triple's validity.},
commutativity, associativity, and distributivity over both finite and infinite unions
(which is crucial to prove soundness of the frame rule, as explained in \secref{subsec:key:embedding-fr}).\footnote{Formal proofs of these properties are available in our Isabelle development.}
Moreover, standard and hyper separating conjunctions are closely related, as shown by the following properties:
\begin{equation}\label{eq:hyper-star-properties}
	\begin{aligned}
		&(\fs{\sigma}\sigma(\epsilon_1:p))\star(\fs{\sigma}\sigma(\epsilon_2:q))\equiv\fs{\sigma}\sigma(\epsilon_1 \bowtie \epsilon_2:p*q)\\
		&(\es{\sigma}\sigma(\epsilon_1:p))\star(\fs{\sigma}\sigma(\epsilon_2:q))\models\es{\sigma}\sigma(\epsilon_1 \bowtie \epsilon_2:p*q).
	\end{aligned}
\end{equation}

\subsection{Relaxing Postconditions with Unknown States}
\label{subsec:formal:uk-res}

As discussed informally in \secref{subsec:key:uk-states}, framing heap-dependent triples expressing reachability, such as \eqref{eqn:non-local-error2}, is generally unsound.
The reason is that frame disjointness is guaranteed solely in the witness state(s); hence, the frame may be altered in other states.
We address this by \emph{relaxing} the postcondition so that the frame applies only to the witness state(s).
To formalize this concept, we begin by defining overapproximations of states and assertions.

As motivated in \secref{subsec:key:uk-states}, unknown states, which come from domain violation errors,
represent our lack of knowledge about corresponding executions in larger states (\ie after a frame has been added).
For example, executing the command $\Free{x}$ in the initial state $\ok: \sigma_u$
(at the top of \figref{subfig:nle:solution}) leads to a domain violation error because the heap location at $x$ is unallocated.
We thus label this state $\uk: \sigma_u$ after the execution of $\Free{x}$.
After adding the frame $\ok: \sigma_f$, we get the state $\uk: \sigma_u \oplus \sigma_f$, which represents our lack of knowledge about the execution of $\Free{x}; C$ in 
the initial state $\ok: \sigma_u \oplus \sigma_f$ (shown at the bottom of \figref{subfig:nle:solution}):
This execution might lead to
an unrelated $\ok$ state,
to an unrelated $\er$ state (\eg if $C$ tries to deallocate a null pointer),
to no state at all (divergence), or to multiple states (non-determinism).

We formally capture this intuition as follows.
We say that a state $\la\Lambda,\sigma_\epsilon\ra$ \emph{overapproximates a state} $\la\Lambda',\sigma'_{\epsilon'}\ra$,
written $\Res[\la\Lambda,\sigma\el\ra][\la\Lambda',\sigma'\elp\ra]$,
iff 
(a) $\epsilon = \uk$ and $\Lambda = \Lambda'$, or
(b) $\epsilon \in \{\ok, \er\}$ and $\la\Lambda,\sigma_\epsilon\ra = \la\Lambda',\sigma'_{\epsilon'}\ra$.
As condition (a) reflects, unknown states overapproximate \emph{any} state sharing the same logical store: while the heap update may lead to arbitrary outcomes, program statements cannot change logical variables.%
\footnote{Preserving logical stores lets us retain logical variables that capture information relevant for relational reasoning (e.g., to distinguish different executions, as discussed by \citet{dardinierHyperHoareLogic2024}).
Alternatively, one could preserve program variables not modified by $C$ and use them in place of logical variables, but this proves too cumbersome in practice.}
In contrast, heap-extended executions for $\ok$ and $\er$ outcomes are fully determined; thus, they only ``overapproximate'' themselves, as reflected by condition (b).

Overapproximation extends naturally to sets of states:
A set $S_0$ \emph{overapproximates a set} $S$, written $\Res[S_0][S]$, iff
(A) all $\ok$ and $\er$ states in $S_0$ are also in $S$, i.e., $\{\la\Lambda,\sigma_\epsilon\ra \in S_0 \mid \epsilon \in \{\ok, \er\}\} \subseteq S$, and
(B)	all states in $S$ are overapproximated by some state in $S_0$, i.e., $\forall \omega \in S \ldotp \exists \omega_0 \in S_0 \ldotp \Res[\omega_0][\omega]$.	
Condition (A) ensures that all known states ($\ok$ and $\er$) in the overapproximating set $S_0$ are also present in the original set $S$,
since their outcomes are fully determined, while condition (B) allows extra states to be accounted for by the overapproximation.

Finally, we define \emph{the
relaxation of} $P$, denoted $\R[P]$, as the set of all sets of states $S$ which are overapproximated by some $S_0 \in P$:
$	\R[P]\triangleq\{S\mid S_0\in P\land \Res[S_0][S] \}$.
Under this relaxation, the framed post of \eqref{eqn:non-local-error2},
$\R[\big(\es{\sigma}\sigma(\ok:x\mapsto\bot)\big)\star\big(\fs{\sigma}\sigma(\ok:\acc{y})\big)]=\es{\sigma}\sigma(\ok:x\mapsto\bot*\acc{y})$,
adds the frame exclusively to the witness state, achieving the intended effect.

Importantly, existential quantification, e.g., $\es{\sigma}\sigma(\ok:\ldots)$, naturally allows $\uk$ states (in contrast to $\fs{\sigma}\sigma(\ok:\ldots)$, which excludes them).
This inclusion of $\uk$ states is essential for the relaxation to soundly overapproximate the program’s non-witnessed behaviors.
For example, \eqref{eqn:non-local-error2} holds because, regardless of which non-witness states appear in the initial set,
the relaxation can soundly overapproximate them via appropriate $\uk$ states, which are admitted by the existential postcondition.

\subsection{Validity of Hyper-Triples}
\label{subsec:formal:hyper-triples}

Intuitively, the hyper-triple $\simpleHoare{P}{C}{Q}$ expresses that for every set $S$ of initial states that satisfies $P$,
the resulting set of final states (after executing $C$ in all states from $S$) satisfies $Q$.
To capture this intuition, we first define the image of a set of states $S$ under the command $C$, denoted $\Sem[C][S]$, as
\begin{align*}
	&\{ \la\Lambda, \sigma'_\ok\ra \mid \la\Lambda, \sigma_\ok\ra \in S \land \la C, \sigma \ra \to^* \la \Skip, \sigma' \ra \} \\
	\cup
	\; &
	\{ \la\Lambda, \sigma'_\er\ra \mid \la\Lambda, \sigma_\ok\ra \in S \land \la C, \sigma \ra \to^* \la C', \sigma' \ra \land \la C', \sigma' \ra \to \cEr \} 
	\cup
	\{ \la\Lambda, \sigma_\er\ra \mid \la\Lambda, \sigma_\er\ra \in S \}
	\\
	\cup
	\; &
	\{ \la\Lambda, \sigma'_\uk\ra \mid \la\Lambda, \sigma_\ok\ra \in S \land \la C, \sigma \ra \to^* \la C', \sigma' \ra \land \la C', \sigma' \ra \to \cDv \}
	\cup
	\{ \la\Lambda, \sigma_\uk\ra \mid \la\Lambda, \sigma_\uk\ra \in S \}
\end{align*}
Successful executions lead to $\ok$ states, whereas
erroneous executions lead to either
$\er$ states (for non-domain violation errors) or
$\uk$ states (for domain violation errors).
Moreover,
$\er$ and $\uk$ states are 
preserved, but
not executed further.
Finally, logical stores are preserved across executions.

Using this definition, we can now define the validity of hyper-triples.
As explained in \secref{sec:key-ideas}, we additionally bake in the preservation of $\forall$-frames (to ensure the soundness of the frame rule),
and relax the postcondition (as explained above) to enable reachability reasoning for heap-manipulating commands.
We thus define the validity of a triple as follows:
$$
\hoare{P}{C}{Q}
\defiff
\textcolor{darkgreen}{\forall f\in\F[\md[C]]\ldotp}
	\forall S \in P \textcolor{darkgreen}{\,\star\,\mathcal{P}(f)} \ldotp
	\Sem[C][S] \in \textcolor{orange}{\R(}Q\textcolor{darkgreen}{\,\star\,\Pow[f]} \textcolor{orange}{)}
$$
The orange part shows the relaxation of the postcondition, as described above.
The green part shows the embedding of the frame rule:
The triple must preserve $\fs{\sigma}\sigma(f)$
(represented by $\Pow[f]$, the powerset of $f$)
for all \emph{admissible} frames $f$,
denoted by $\F[\md[C]]$.
A frame is \emph{admissible}
iff (1) $\fv[f]\cap\md[C]=\emptyset$ holds%
\footnote{Formally corresponding to
$\forall\Lambda,s,s',h,\epsilon\ldotp(\forall x\notin\md[C]\ldotp s(x)=s'(x))\Longrightarrow(\la\Lambda,\la s,h\ra\el\ra\in f\Longleftrightarrow\la\Lambda,\la s',h\ra\el\ra\in f)$.},
and (2) $f$ is satisfied by $\ok$ states only.
Condition (1) is standard, and condition (2)
is necessary, since framing with erroneous frames may prevent the execution of $C$:
$\hoare{\fs{\sigma}\sigma(\ok:x=0)}{\Assign{x}{5}}{\fs{\sigma}\sigma(\ok:x=5)}$ holds, but framing it with
$\fs{\sigma}\sigma(\er:\top)$ prevents all executions of $\Assign{x}{5}$, so the postcondition no longer accurately describes the outcome.

\subsection{Rules}
\label{subsec:formal:rules}

\begin{figure}
	\begingroup
	\newcommand{\scaleabortfig}{0.7}
	\scalebox{\scaleabortfig}{$\newrule{Frame}{rule:frame}[\text{no intersecting scaffold variables in }Q\text{ and }F][\hoare{P}{C}{Q}][F\models\fs{\sigma}\sigma(\ok:\top)][\fv[F]\cap\md[C]=\emptyset][\text{no }\exists\la\bm\cdot\ra\text{ in }F]{\hoare{P\star F}{C}{Q\star F}}$}
	\hspace{2pt}
	\scalebox{\scaleabortfig}{$\newrule{JoinTrue}{rule:join-true}[\hoare{P}{C}{Q}]{\hoare{P\otimes\top}{C}{Q\otimes\top}}$}
	\hfill
	\scalebox{\scaleabortfig}{$\newrule{Seq}{rule:seq}[\hoare{P}{C_1}{R}][\hoare{R}{C_2}{Q}]{\hoare{P}{C_1;C_2}{Q}}$}
	\\[\bigskipamount]
	\scalebox{\scaleabortfig}{$\newrule{Alloc}{ax:alloc}[\baseok[P]][x\notin\fv[P]]{\hoare{P}{x\coloneqq\mathrm{alloc}()}{(\fs{\sigma}\sigma(\ok:x\mapsto\sth))\star P}}$}\hfill 
	\scalebox{\scaleabortfig}{$\newrule{Read}{ax:read}[\baseok[P]][y\notin\fv[P]\cup\const{pvars}(e)\cup\{x\}]{\hoare{(\fs{\sigma}\sigma(\ok:x\mapsto e))\star P}{y\coloneqq[x]}{(\fs{\sigma}\sigma(\ok:x\mapsto e\land y=e))\star P}}$} 
	\\[\bigskipamount]
	\scalebox{\scaleabortfig}{$\newrule{Write}{ax:write}[\baseok[P]]{\hoare{(\fs{\sigma}\sigma(\ok:x\mapsto\sth))\star P}{\Write{x}{e}}{(\fs{\sigma}\sigma(\ok:x\mapsto e))\star P}}$}\hfill 
	\scalebox{\scaleabortfig}{$\newrule{Free}{ax:free}[\baseok[P]]{\hoare{(\fs{\sigma}\sigma(\ok:x\mapsto\sth))\star P}{\Free{x}}{(\fs{\sigma}\sigma(\ok:x\mapsto\bot))\star P}}$} 
	\\[\bigskipamount]
	\scalebox{\scaleabortfig}{$\newrule{Cons}{rule:cons}[\hoare{P}{C}{Q}][P'\subseteq P][Q\subseteq Q']{\hoare{P'}{C}{Q'}}$}\hfill 
	\scalebox{\scaleabortfig}{$\newrule{WhileSync}{rule:while-sync}[I\models\fs{\sigma_1}\fs{\sigma_2}\sigma_1(\ok:b)\Leftrightarrow\sigma_2(\ok:b)][I\models(\fs{\sigma}\sigma(\ok:\top))\lor(\fs{\sigma}\sigma(\er:\top))][\hoare{(\fs{\sigma}\sigma(\ok:b))\land I}{C}{I}]{\hoare{I}{\while{b}{C}}{((\fs{\sigma}\sigma(\ok:\lnot b))\land I)\lor((\fs{\sigma}\sigma(\er:\top))\land I)\lor(\fs{\sigma}\bot)}}$}
	\vspace{-0.2cm}
	\caption{Selected Rules of Hyper Separation Logic.} \label{fig:selected-rules}
	\vspace{-0.1cm}
	\endgroup
\end{figure}

We present a selection of HSL's rules in \figref{fig:selected-rules}, focusing on rules specific to our separation logic setting.
Many more HSL rules,
including those to reason about (local) errors,
are shown in \appref{sec:appendix-soundness}.
The rule \ruleref{rule:frame}, crucial for local reasoning, has already been illustrated in \secref{subsec:frame-rule-tour} and 
discussed in \secref{subsec:key:embedding-fr} and above.
Its last restriction, forbidding $\exists\la\bm\cdot\ra$ in the frame $F$, is necessary for soundness because reachability in the postcondition implies termination, while $C$ might not always terminate.
The scaffold restriction is syntax-specific, explained in \secref{subsec:formal:syntax}.

The four rules for basic heap-manipulating commands (\ruleref{ax:alloc}, \ruleref{ax:read}, \ruleref{ax:write}, and \ruleref{ax:free}) are uniformly formulated to internalize a \emph{strengthened} (in two ways) frame rule.
First, unlike the frame $F$ in \ruleref{rule:frame}, $P$ can existentially quantify over states, as these commands never diverge.
Second, the \emph{syntactic} restriction $\baseok[P]$, which prevents $P$ from mentioning $\er$ labels, is weaker and easier to check than the \emph{semantic} $\ok$ only restriction of \ruleref{rule:frame}. 
As a result, these rules may be applied with (internalized) frames $P$ such as $\es{\sigma}\sigma(\ok:\top)$,
while the \ruleref{rule:frame} rule rejects such a frame in two ways: it contains $\exists\la\bm\cdot\ra$ and allows non-$\ok$ states.
As illustrated in \secref{sec:tour}, the corresponding $\forall$-conjunct can be distributed over quantified states,
e.g., $(\fs{\sigma}\sigma(\ok:p)) \star (\fs{\sigma_1} \es{\sigma_2} \exists v \ldotp \sigma_1(\ok:q_1(v)) \land \sigma_2(\ok:q_2(v)))$
is equivalent to $\fs{\sigma_1} \es{\sigma_2} \exists v \ldotp \sigma_1(\ok:p * q_1(v)) \land \sigma_2(\ok:p * q_2(v))$.

To apply any of these rules, one must first rewrite the precondition in the form $(\fs{\sigma}\sigma(\ok:p))\star P$.
As we have seen in \secref{subsec:hyper-star-tour}, this rewriting is straightforward for $\forall^+ \exists^*$-preconditions, as all states are constrained, but is more involved
for $\exists^*$-preconditions.
For example, consider applying the rule \ruleref{ax:free} to derive triple \eqref{eqn:non-local-error2} from \secref{sec:key-ideas}:
We cannot rewrite directly
our precondition $\es{\sigma}\sigma(\ok:\acc{x})$
as $(\fs{\sigma}\sigma(\ok:\acc{x}))\star P$ for some $P$,
because $\acc{x}$ does not hold in \emph{all} states, only in the witness state.
To circumvent this issue,
we first rewrite $\es{\sigma}\sigma(\ok:\acc{x})$ as
$(\es{\sigma}\sigma(\ok:\acc{x})) \otimes \top$:
Intuitively, $P \otimes \top$ expresses that only a \emph{subset} of states has to satisfy $P$.%
\footnote{Here, $\otimes$ is the join operation $P\otimes Q\triangleq\{S_P\cup S_Q\mid S_P\in P\land S_Q\in Q\}$.}
Since $\otimes \top$ allows us to discard states,
we now have the equivalence with
$((\fs{\sigma} \sigma(\ok:\acc{x})) \star (\es{\sigma}\sigma(\ok:\top))) \otimes \top$:
We discard all states that do not satisfy $\acc{x}$,
retaining only those that do (including the witness).
Finally, we can eliminate the $\otimes \top$ using the \ruleref{rule:join-true} rule, as follows: 

{\scriptsize
$$
\applyrule{rule:join-true}
[
\applyrule{ax:free}[~]{%
\hoare{(\fs{\sigma}\sigma(\ok:\acc{x}))\star (\es{\sigma}\sigma(\ok:\top))}
{\Free{x}}
{(\fs{\sigma}\sigma(\ok:\pointstobot{x}))\star (\es{\sigma}\sigma(\ok:\top))}
}
]
{
\hoare%
{\underbrace{\Big((\fs{\sigma}\sigma(\ok:\acc{x}))\star (\es{\sigma}\sigma(\ok:\top))
\Big)\otimes\top}_{\es{\sigma}\,\sigma(\ok:\,\acc{x})}}%
{\Free{x}}%
{\underbrace{\Big((\fs{\sigma}\sigma(\ok: \pointstobot{x}))\star (\es{\sigma}\sigma(\ok:\top))\Big)\otimes\top}_{\es{\sigma}\,\sigma(\ok:\,\pointstobot{x})}}
}
$$}
The same pattern can be applied to other heap-manipulating commands when the precondition is not in the desired form,
and to use the rule \ruleref{rule:frame} with $\exists^*$-preconditions and postconditions.
While handling such cases may seem complex at first, these steps are largely routine and could be automated or incorporated into derived rules, so they do not pose a significant practical challenge.

The rule \ruleref{rule:while-sync} is analogous to that of \citet{dardinierHyperHoareLogic2024}, but generalized to account for local errors.
It applies to loops whose executions preserve identical control flow across all runs.
Note that $\sigma_1(\ok:b)\Leftrightarrow\sigma_2(\ok:b)$ is short for $(\sigma_1(\ok:b)\Rightarrow\sigma_2(\ok:b))\land(\sigma_2(\ok:b)\Rightarrow\sigma_1(\ok:b))$,
where $\sigma_1(\ok:b)\Rightarrow\sigma_2(\ok:b)$ is short for $\sigma_1(\ok:\lnot b)\lor\sigma_1(\er:\top)\lor\sigma_1(\uk:\top)\lor\sigma_2(\ok:b)$.

\subsection{Soundness}
\label{subsec:formal:soundness}

We have formalized and proven in Isabelle/HOL that Hyper Separation Logic is sound:
\begin{theorem}\label{th:soundness}
	The rules of Hyper Separation Logic, presented in \figref{fig:selected-rules} and in \appref{sec:appendix-soundness},
	 are sound with respect to the definition of validity of hyper-triples (\secref{subsec:formal:hyper-triples}).
\end{theorem}
We discuss the proof in \appref{sec:appendix-soundness}.
Moreover, while our definition of validity of hyper-triples is rather involved (since we embed all $\forall$-frames and relax the postcondition), we prove the following adequacy theorem, showing that hyper-triples imply the expected behavior:
\begin{theorem}\label{th:adequacy}
	If $\hoare{P}{C}{Q}$ holds, and the set $S$ of initial states satisfies $P$ (\ie $S\models P$), then:
	\begin{enumerate}
		\item The set $\Sem[C][S]$ of reachable states satisfies the relaxation of $Q$ (\ie $\Sem[C][S] \models \R[Q]$).
		\item Additionally, if $Q$ does not mention any $\uk$ label, contains no $\star$, and constrains its existentially-quantified states to be $\ok$ or $\er$,%
\footnote{This last restriction is required because
$\R[\es{\sigma} \sigma(\ok: p)] = \es{\sigma} \sigma(\ok: p)$, but $\R[\es{\sigma} \top] = \top$.}
		then $\Sem[C][S]$ satisfies $Q$ (\ie $\Sem[C][S] \models Q$).
	\end{enumerate}
\end{theorem}
As explained in \secref{subsec:formal:state-model}, we obtain (1) by eliminating the embedded frame using the unit element of $\star$.
Point (2)
shows that
HSL is precise for postconditions that do not mention $\uk$ labels,
\ie for hyperproperties involving successful and erroneous executions (for errors such as accesses to freed memory, null-pointer dereferences, or failed assertions); the $\star$ can usually%
\footnote{Similarly to $*$, $\star$ does not generally distribute over $\land$.} be eliminated by distributing it over the different connectives
(\eg as done in \secref{subsec:hyper-star-tour}).

\subsection{Syntactic Hyper-Assertions and Scaffold Variables}
\label{subsec:formal:syntax}

To ease reasoning,
we have formalized the following syntax for hyper-assertions:
$$
P
\Coloneqq
\top\mid\bot\mid\sigma(\epsilon: p)\mid\en{n}P\mid\fn{n}P\mid\es{\sigma}P\mid\fs{\sigma}\textbf{}P\mid P\lor P\mid P\land P\mid P\otimes P\mid P\star P
$$
where $P$ ranges over (syntactic) hyper-assertions,
$\sigma$ over state variables in $\SVars$,
$\epsilon$ over labels $\{\ok, \er, \uk\}$,
$p$ over standard separation logic assertions with scaffold variables (as we explain next),
and $n$ over program variables%
\footnote{For simplicity, we represent those as De Bruijn indices~\cite{debruijnLambdaCalculusNotation1972} in our formalization.}.
We define their interpretation in \appref{sec:appendix},
and 
prove \citep{HyperSeparationLogic_artifact}
many of the equivalences and entailments used throughout the paper,
such as
$(\fs{\sigma} P) \star (\fs{\sigma} Q) \models \fs{\sigma} (P \star Q)$,
$(\fs{\sigma} P) \star (\es{\sigma} Q) \models \es{\sigma} (P \star Q)$,
or $\sigma(\epsilon_1: p) \star \sigma(\epsilon_2: q) \equiv \sigma(\epsilon_1 \bowtie \epsilon_2: p * q)$;
we show more in \appref{sec:appendix}.
We have also used this syntax to formalize the syntactic requirements for different rules,
for example the $\baseok[P]$ condition in the rules \ruleref{ax:alloc}, \ruleref{ax:read}, \ruleref{ax:write}, and \ruleref{ax:free},
the absence of $\exists\la\bm\cdot\ra$ in the frame $F$ of the \ruleref{rule:frame} rule,
or the suitability condition for the second point of the adequacy theorem.%
\footnote{In practice, we prove the soundness of these rules under a suitable \emph{semantic} restriction, implied by the \emph{syntactic} ones.}
All rules consider closed hyper-assertions (no free state variables).

\paragraph{Scaffold variables}
Intuitively, scaffold variables can be thought of as existentially-quantified logical variables \emph{per state}:
A set of states satisfies a hyper-assertion with scaffold variables iff there exists a way to assign values to the scaffold variables \emph{in each state} such that the hyper-assertion holds for the resulting set of states.
We formalize this notion in \appref{sec:appendix}.
Scaffold variables allow hyper-assertions to be conveniently split into a \emph{non-relational} part, constraining individual states, and a \emph{relational} part, relating multiple states.
For example, the postcondition of the function $\mathit{compute}$  from \secref{sec:tour} can be rewritten
with \emph{scaffold variables} $\delta_o$ and $\delta_h$ as follows
(where $e_1(\sigma_1) \mathop{\Rel} e_2(\sigma_2)$, is shorthand for $\en{n}\sigma_1(e_1=n) \land \sigma_2(n\mathop{\Rel}e_2)$, where $\mathop{\Rel}\in\{=,\neq,<,\ldots\}$):
\begin{align*}
&\fs{\sigma_1} \fs{\sigma_2} \es{\sigma}
\exists u, v \ldotp
\sigma_1(\acc{o} * \pointsto{h}{v})
\land
\sigma(\pointsto{o}{u} * \pointsto{h}{v})
\land
\sigma_2(\pointsto{o}{u} * \acc{h})\\
=
&\underbrace{(\fs{\sigma}\sigma(o\mapsto\greenify{\delta_o}*h\mapsto\greenify{\delta_h}))}_{\text{non-relational part}}
\star
\underbrace{(\fs{\sigma_1}\fs{\sigma_2}\es{\sigma}
\sigma(\greenify{\delta_h}) = \sigma_1(\greenify{\delta_h})
\land
\sigma_2(\greenify{\delta_o}) = \sigma(\greenify{\delta_o}))}_{\text{relational part}}
\end{align*}
This formulation
has the advantage that ownership assertions need not be duplicated for $\sigma$, $\sigma_1$, and $\sigma_2$. 
The \emph{scaffold variables} $\delta_o$ and $\delta_h$ represent the values stored at locations $o$ and $h$, respectively. 
Their values are preserved by the hyper separating conjunction $\star$, which allows the second conjunct to refer to those values to express the relational property.
Not only do scaffold variables allow writing more concise specifications, they also enable stronger reasoning principles,
such as:
$$
\newrule{ReadScf}{ax:read-scaffold}
[\baseok[P]]
[y\notin\const{fv}(P)]
[y \neq x]
{\hoare{(\fs{\sigma}\sigma(\ok:x\mapsto\delta_x))\star P}{y\coloneqq[x]}{(\fs{\sigma}\sigma(\ok:x\mapsto\delta_x\land y=\delta_x))\star P}}
$$

Rule \ruleref{ax:read-scaffold} generalizes \ruleref{ax:read} from \figref{fig:selected-rules} by allowing $x$ to hold an abstract value $\delta_x$, constrained in $P$,
rather than requiring it to point to a specific expression $e$, making it more expressive.
As an example, the triple below
can be derived from 
the consequence rule and
\ruleref{ax:read-scaffold}
(with $P \triangleq \fs{\sigma_1} \es{\sigma_2} \exists n \ldotp \sigma_1(\delta_x = n) \land \sigma_2(\delta_x = n+1)$),
but it cannot be derived from \ruleref{ax:read}:
{
\small
$$
\simpleHoare{\fs{\sigma_1} \es{\sigma_2} \exists n \ldotp \sigma_1(x \mapsto n) \land \sigma_2(x \mapsto n+1)}{y \coloneqq [x]}{\fs{\sigma_1} \es{\sigma_2} \exists n \ldotp \sigma_1(y = n) \land \sigma_2(y = n+1)}
$$%
}

Finally, since scaffold variables play a role solely during the interpretation process---serving only as an auxiliary “scaffold”---their values are not required to be preserved from the precondition to the postcondition.
This motivates the scaffold restriction in the rule \ruleref{rule:frame}, as otherwise:
\small
\[\applyrule{rule:frame}[\hoare{\fs{\sigma}\sigma(x\mapsto\delta_x)}{\Write{x}{42}}{\fs{\sigma}\sigma(x\mapsto\delta_x)}]{\nothoare{(\fs{\sigma}\sigma(x\mapsto\delta_x))\star(\fs{\sigma}\sigma(\delta_x=5))}{\Write{x}{42}}{(\fs{\sigma}\sigma(x\mapsto\delta_x))\star(\fs{\sigma}\sigma(\delta_x=5))}}\]
\normalsize

\subsection{Examples}

To conclude this section, we illustrate the expressiveness of HSL’s rules by showing that the program in \figref{fig:violation-GNI} (in black) violates generalized non-interference (GNI), an $\exists\exists\forall$-hyperproperty,
while the program in \figref{fig:example-GNI} satisfies GNI ($\forall \forall \exists$), both lying beyond the reach of existing separation logics.
We provide further examples in \appref{sec:appendix-main-examples}, including a proof of the existence of an execution with a maximal value ($\exists \forall$): another property that lies beyond the capabilities of existing separation logics.

The heap-independent program in the middle of \figref{fig:violation-GNI} is borrowed from Hyper Hoare Logic~\cite{dardinierHyperHoareLogic2024};
we thus focus on the surrounding heap-manipulating commands.
We assume that the high input is stored at location $h$,
and that we have at least two different possibilities for that input value, as expressed by the precondition (where $\delta_h$ is a scaffold variable).
We first use the rule \ruleref{ax:read-scaffold} to read the value pointed by $h$ into variable $v_h$.
Note that the rule \ruleref{ax:read} cannot be used directly here, as we do not have an expression $e$ to which $h$ points in the precondition.
We then rewrite our hyper-assertion (using the consequence rule) before the deallocation,
to apply the rule \ruleref{ax:free}.
After the heap-independent part,
which can be easily handled with the relevant rules from \appref{sec:appendix-soundness} (as shown by \citet{dardinierHyperHoareLogic2024}),
we apply the rules \ruleref{ax:alloc} and \ruleref{ax:write} in sequence, yielding the negation of GNI:
All executions starting with the same high input as $\sigma_1$ end up with a low output different from that of $\sigma_2$.
In other words, observing $\sigma_2$'s low output leaks information about $\sigma_2$'s high input, namely that it is different from $\sigma_1$'s high input.

\begin{figure}
	\begin{minipage}{1\linewidth} 
		\tiny
		\begin{align*}
			&\hasrt{ (\fs{\sigma}\sigma(\pointsto{h}{\delta_h}) ) \star \big((\es{\sigma_1} \es{\sigma_2} \sigma_1(\delta_h) \neq \sigma_2(\delta_h))\land(\fs{\sigma}\sigma(\ok:\top))\big)} \\
			&\Read{v_h}{h}; \tag{\ruleref{ax:read-scaffold}}\\
			&\hasrt{ (\fs{\sigma}\sigma(\pointsto{h}{\delta_h} \land v_h = \delta_h) ) \star \big((\es{\sigma_1} \es{\sigma_2} \sigma_1(\delta_h) \neq \sigma_2(\delta_h))\land(\fs{\sigma}\sigma(\ok:\top))\big)}\\
			\models &\hasrt{ (\fs{\sigma}\sigma(\pointsto{h}{\sth}) ) \star \big((\es{\sigma_1} \es{\sigma_2} \sigma_1(v_h) \neq \sigma_2(v_h))\land(\fs{\sigma}\sigma(\ok:\top))\big)} \\
			&\Free{h}; \tag{\ruleref{ax:free}}\\
			&\hasrt{ (\fs{\sigma}\sigma(\pointstobot{h}) ) \star \big((\es{\sigma_1} \es{\sigma_2} \sigma_1(v_h) \neq \sigma_2(v_h))\land(\fs{\sigma}\sigma(\ok:\top))\big)}\\
			\models &\hasrt{\es{\sigma_1}\en{k_1}\sigma_1(k_1\le9)\land \es{\sigma_2}\en{k_2}\sigma_2(k_2\le9)\land \fs{\sigma}\fn{k}\sigma(k\nleq9)\lor \sigma(v_h) \neq \sigma_1(v_h) \lor \sigma(v_h+k) \neq \sigma_2(v_h+k_2)} \\
			&\Havoc{k};
			\Assume{k \le 9};
			\Assign{v_l}{v_h + k}; \\
			&\hasrt{ \es{\sigma_1} \es{\sigma_2} \fs{\sigma} \sigma(v_h) \neq \sigma_1(v_h) \lor \sigma(v_l) \neq \sigma_2(v_l)} \\
			&\Alloc{l}; \tag{\ruleref{ax:alloc}}\\
			&\hasrt{ (\fs{\sigma} \sigma(\pointsto{x}{\sth})) \star (\es{\sigma_1} \es{\sigma_2} \fs{\sigma} \sigma(v_h) \neq \sigma_1(v_h) \lor \sigma(v_l) \neq \sigma_2(v_l))}\\
			&\Write{l}{v_l} \tag{\ruleref{ax:write}}\\
			&\hasrt{ (\fs{\sigma} \sigma(\pointsto{x}{v_l})) \star (\es{\sigma_1} \es{\sigma_2} \fs{\sigma} \sigma(v_h) \neq \sigma_1(v_h) \lor \sigma(v_l) \neq \sigma_2(v_l))}
		\end{align*}
		\vspace{-0.5cm}
		\caption{Proof outline showing that the program in black violates GNI, where the high input is stored at location $h$, and the low output is stored at the newly-allocated location $l$.
			A proof outline for the heap-independent part can be seen in \citep{dardinierHyperHoareLogic2024}.
		}
		\label{fig:violation-GNI}
	\end{minipage}
\end{figure}

\begin{figure}
	\begin{minipage}{\linewidth} 
		\tiny
		\begin{align*}
			\hspace{0pt}
			&\hasrt{(\fs{\sigma}\sigma(\List[h][H]))\star(\fs{\sigma_1}\fs{\sigma_2} \sigma_1(\len(H))=\sigma_2(\len(H)))} \\
			\models &\hasrt{(\fs{\sigma}\sigma(\List[0][\delta_L]*\Lseg[h][h][\delta_A]*\List[h][\delta_B]*H={\delta_A}^\frown\delta_B))\star
				I_{rel}} \\
			&\Assign{i}{h};\Assign{l}{0};\Assign{s}{0}; \\
			&\hasrt{I_{own}\star
				I_{rel}} \\ 
			&\resv{while }i\neq0\resv{ do} \\
			&\;\hasrt{(\fs{\sigma}\sigma(i\neq0))\land(I_{own}\star
				I_{rel})} \\
			\models &\;\hasrt{(\fs{\sigma}\sigma(i\mapsto\delta))\star
				(\fs{\sigma}\sigma(\List[l][\delta_L]*\Lseg[h][i][\delta_A]*\delta=\tb{\delta_v,\delta_{i'}}*\List[\delta_{i'}][\delta_B]*H={\delta_A}^\frown[\delta_v]^\frown\delta_B))\star
				I_{rel}} \\
			&\;\Read{p}{i}; \tag{\ruleref{ax:read-scaffold}}\\
			&\;\hasrt{(\fs{\sigma}\sigma(i\mapsto\delta\land p=\delta))\star
				(\fs{\sigma}\sigma(\List[l][\delta_L]*\Lseg[h][i][\delta_A]*\delta=\tb{\delta_v,\delta_{i'}}*\List[\delta_{i'}][\delta_B]*H={\delta_A}^\frown[\delta_v]^\frown\delta_B))\star
				I_{rel}} \\
			\models &\;\hasrt{(\fs{\sigma}\sigma(\List[l][\delta_L]*\Lseg[h][p\ldotp\text{snd}][\delta_A]*\List[p\ldotp\text{snd}][\delta_B]*H={\delta_A}^\frown\delta_B))\star
				I_{rel}} \\
			&\;\Assign{s}{s+p\ldotp\text{fst}};\Assign{i}{p\ldotp\text{snd}}; \\
			&\;\hasrt{(\fs{\sigma}\sigma(\List[l][\delta_L]*\Lseg[h][i][\delta_A]*\List[i][\delta_B]*H={\delta_A}^\frown\delta_B))\star
				I_{rel}} \\ 
			\models &\;\hasrtl{(\fs{\sigma}\fn{k}\sigma(\List[l][\delta_L]*\Lseg[h][i][\delta_A]*\List[i][\delta_B]*H={\delta_A}^\frown\delta_B))\,\star} \\
			&\;\hasrtr{\ \ (\fs{\sigma_1}\fn{k_1}\fs{\sigma_2}\fn{k_2}\es{\sigma}\en{k} \sigma_1(\len[\delta_B]) = \sigma_2(\len[\delta_B]) \land \sigma(H)=\sigma_1(H)\land
				\sigma([s\oplus k]^\frown\delta_L) = \sigma_2([s\oplus k_2]^\frown\delta_L))} \\
			&\;\Havoc{k}; \\
			&\;\hasrt{I_{own}\star
				(\fs{\sigma_1}\fs{\sigma_2}\es{\sigma} \sigma_1(\len[\delta_B]) = \sigma_2(\len[\delta_B]) \land \sigma(H) = \sigma_1(H)\land \sigma([s\oplus k]^\frown\delta_L) = \sigma_2([s\oplus k]^\frown\delta_L))} \\
			&\;\Alloc{p}; \tag{\ruleref{ax:alloc}}\\
			&\;\hasrt{(\fs{\sigma}\sigma(p\mapsto\_))\star I_{own}\star(\fs{\sigma_1}\fs{\sigma_2}\es{\sigma} \sigma_1(\len[\delta_B]) = \sigma_2(\len[\delta_B]) \land \sigma(H) = \sigma_1(H)\land \sigma([s\oplus k]^\frown\delta_L) = \sigma_2([s\oplus k]^\frown\delta_L))} \\
			&\;\Write{p}{\tb{s\oplus k,l}}; \tag{\ruleref{ax:write}}\\
			&\;\hasrt{(\fs{\sigma}\sigma(p\!\mapsto\!\!\tb{s\oplus k,l}))\!\star\! I_{own}\!\star\!
				(\fs{\sigma_1}\!\fs{\sigma_2}\!\es{\sigma} \sigma_1(\len[\delta_B])\!=\!\sigma_2(\len[\delta_B])\!\land\!\sigma(H)\!=\!\sigma_1(H)\!\land\!\sigma([s\!\oplus\! k]^\frown\delta_L)\!=\!\sigma_2([s\!\oplus\! k]^\frown\delta_L))} \\
			\models &\;\hasrt{(\fs{\sigma}\sigma(\List[p][\delta_L]*\Lseg[h][i][\delta_A]*\List[i][\delta_B]*H={\delta_A}^\frown\delta_B))\star
				I_{rel}} \\
			&\;\Assign{l}{p}; \\
			&\;\hasrt{I_{own}\star
				I_{rel}} \\ 
			&\resv{od} \tag{\ruleref{rule:while-sync}}\\
			&\hasrt{\big((\fs{\sigma}\sigma(i=0))\land(I_{own}\star I_{rel})\big)\lor\big((\fs{\sigma}\sigma(\er:\top))\land(I_{own}\star I_{rel})\big)\lor(\fs{\sigma}\bot)} \\
			\models &\hasrt{(\fs{\sigma}\sigma(\List[l][\delta_L]*\List[h][H]))\star(\fs{\sigma_1}\fs{\sigma_2}\es{\sigma} \sigma(H) = \sigma_1(H)\land \sigma(\delta_L) = \sigma_2(\delta_L))} \\
		\end{align*}
		\vspace{-0.8cm}
		\caption{Proof that the program in black satisfies GNI, with the (unaltered) $H$ denoting mathematical sequence, being secret, while its length is public.
				The predicate $\List[h][H]$ is defines as $\Lseg[h][0][H]$,
				where the predeicate $\Lseg[h][h_{int}][H]$ is defined as $(h=h_{int}\,\land\, H=[])\lor(\en{a,h',H'}h\mapsto\la a,h'\ra\,*\,\Lseg[h'][h_{int}][H']\,*\,H=[a]^\frown H')$.
				The loop invariant is given by
				$I_{own}\triangleq\fs{\sigma}\sigma(\List[l][\delta_L]*\Lseg[h][i][\delta_A]*\List[i][\delta_B]*H={\delta_A}^\frown\delta_B)$, separation-conjoined with
				$I_{rel}\triangleq\fs{\sigma_1}\fs{\sigma_2}\es{\sigma} \sigma_1(\len[\delta_B]) = \sigma_2(\len[\delta_B]) \land \sigma(H) = \sigma_1(H)\land \sigma(\delta_L) = \sigma_2(\delta_L)$.
		}
		\label{fig:example-GNI}
	\end{minipage}%
\end{figure}

Similarly to the proof in \figref{fig:violation-GNI}, the proof in \figref{fig:example-GNI} focuses on the heap-manipulating commands,
as the heap-independent aspects are easily established with the relevant rules from \appref{sec:appendix-soundness} (cf. \citep{dardinierHyperHoareLogic2024} for related examples).
We assume that the high input is stored in the contents of the list at $h$ (represented by the mathematical sequence $H$), while its length is treated as a low-security input.
After applying the consequence rule and initializing $\Assign{i}{h}$ as a traverse pointer, $\Assign{l}{0}$ for the new list under construction,
and $\Assign{s}{0}$ to accumulate the sum of traversed secret values, we ensure that the invariant $I_{own}\star I_{rel}$ holds at entry.

The proof of the loop invariant begins with the standard ownership factorization, followed by the application of \ruleref{ax:read-scaffold} (similarly to \figref{fig:violation-GNI}, we cannot apply \ruleref{ax:read}).
We then leverage the fact that scaffold variables are not preserved from pre- to postcondition, renaming $\delta_A$ to include the previous $\delta_A$ together with the iterated element $\delta_v$.
Next, we systematically apply the assign and havoc rules from \appref{sec:appendix-soundness}, after which \ruleref{ax:alloc} is applied.
This produces a hyper-assertion in the required form, allowing a straightforward application of \ruleref{ax:write}.
The invariant proof concludes with a renaming of $\delta_L$ and an application of the assign rule from \appref{sec:appendix-soundness}.

Since the invariant ensures that the loops progress in sync, we can apply the \ruleref{rule:while-sync} rule.
Its erroneous disjunct, $(\fs{\sigma}\sigma(\er:\top))\land(I_{own}\star I_{rel})$, can be dropped as the corresponding invariant is $\ok$-only, while
its third disjunct, $\fs{\sigma}\bot$, can be dropped as the invariant is downward closed.
We conclude the proof with a final application of \ruleref{rule:cons} together with the definitions of $\List$ and $\Lseg$.

\section{Related Work}
\label{sec:related}
As discussed throughout the paper, the closest related works are Hyper Hoare Logic (HHL)~\cite{dardinierHyperHoareLogic2024} and Outcome Separation Logic (OSL)~\cite{zilbersteinOutcomeSeparationLogic2024}.
Our work builds on the ideas from HHL, in particular, tracking sets of reachable states and using universal and existential quantification over initial and reachable states to express and prove hyperproperties. However, HSL goes significantly beyond HHL by supporting separation logic reasoning about heap-manipulating programs, which is enabled by its novel hyper separating conjunction and generalized frame rule. Both are substantial generalizations of their standard SL counterparts.

Outcome Separation Logic (OSL)~\cite{zilbersteinOutcomeSeparationLogic2024} extends Outcome Logic (OL)~\cite{zilbersteinOutcomeLogicUnifying2023} with separation-logic-style correctness and incorrectness reasoning for heap-manipulating programs.
Similarly to HSL, OSL's judgments (in its non-deterministic instantiation) are interpreted over sets of states, and thus both logics face similar challenges.
In particular, OSL can express combinations of unary safety and reachability properties, \ie combinations of $\forall$- and $\exists$-properties.
However, unlike HSL, OSL cannot express or reason about hyperproperties, \ie properties relating multiple executions of a program, such as monotonicity and non-interference ($\forall\forall$), GNI ($\forall\forall\exists$), or non-determinism ($\exists\exists$), since these require more than one state quantifier.
This difference is also reflected in the design of the separating conjunction and frame rule.
A key technical feature of OSL is an \emph{asymmetric} separating conjunction, defined \emph{syntactically}, between an outcome assertion and a standard unary SL assertion.
Correspondingly, the frame in OSL's frame rule is restricted to a standard SL assertion $f$, which intuitively corresponds to the HSL assertion $\forall\langle \sigma \rangle.\ \sigma(\mathit{ok}: f)$.
If adopted in our setting, such an asymmetric syntactic star would restrict framing to unary assertions, and would therefore prevent framing relations between executions.
In contrast, HSL uses a \emph{symmetric} hyper separating conjunction, defined \emph{semantically} over sets of labeled states.
Accordingly, HSL's frame rule allows $\forall^*$-frames expressing relational properties,
which is crucial for compositional reasoning about hyperproperties, and provides an important foundation for future extensions to concurrency (where the postconditions of parallel threads are themselves hyper-assertions).
Finally, OSL models domain violations using a dedicated $\mathit{undef}$ state and a special outcome assertion $\top$, whereas HSL tracks unknown states via a $\uk$ label; in our setting, this design is essential to obtain an associative symmetric hyper separating conjunction.

Similarly to HSL and OSL,
many logics for reachability have been extended from Hoare logic~\cite{floydAssigningMeaningsPrograms1967, hoareAxiomaticBasisComputer1969} versions to separation logic (SL)~\cite{reynoldsSeparationLogicLogic2002} versions, to support local reasoning for heap-manipulating programs, including
(Concurrent) Incorrectness Separation Logic (ISL)~\cite{raadLocalReasoningPresence2020,raadConcurrentIncorrectnessSeparation2022} ($\exists$-properties, based on Incorrectness Logic~\cite{ohearnIncorrectnessLogic2019}),
Separation Sufficient Incorrectness Logic (SSIL)~\cite{ascariSufficientIncorrectnessLogic2024} ($\exists$-properties), and
Exact Separation Logic (ESL)~\cite{maksimovicExactSeparationLogic2023} ($\forall$ and $\exists$-properties).
Unlike HSL, these logics reason only about \emph{unary} properties of programs,
and thus defining the separating conjunction in their settings is straightforward.
Similarly, \emph{InsecSL}~\cite{murrayCompositionalVulnerabilityDetection2023} extends
\emph{Insecurity Logic}~\cite{murrayUnderApproximateRelationalLogic2020},
a program logic for proving violations of $\forall \forall$-properties (\ie $\exists \exists$-properties),
to a separation logic setting.
HSL subsumes SL ($\forall$) and SSIL ($\exists$), in the sense that the triples derivable in these logics are also derivable in HSL.
Moreover, ISL ($\exists$), ESL ($\forall$ and $\exists$), and Insec triples ($\exists \exists$) can be expressed as HSL triples, but not with the standard state-quantifiers ($\forall \langle \sigma \rangle$ and $\exists \langle \sigma \rangle$), as shown in \citet{dardinierHyperHoareLogic2024_extended}.
Because of this discrepancy, we have not explored whether their rules can be derived from HSL rules.
Additionally, HSL can prove more general reachability properties, such as GNI, that involve both existential and universal quantification; it also allows one to compose different kinds of properties in one proof, for instance, to compose non-interference for a deterministic statement with GNI for a non-deterministic one.

In the domain of safety hyperproperties,
Relational Separation Logic (RSL)~\cite{yangRelationalSeparationLogic2007}
extends Relational Hoare Logic~\cite{bentonSimpleRelationalCorrectness2004} to reason about $\forall \forall$-properties of heap-manipulating programs,
and LGTM~\cite{gladshteinMechanisedHypersafetyProofs2024} extends LHC~\cite{dosualdoProvingHypersafetyCompositionally2022} (itself extending Cartesian Hoare Logic~\cite{sousaCartesianHoareLogic2016}) to reason about $\forall^*$ properties of heap-manipulating programs.
These logics allow proving \emph{relational properties}, \ie
properties relating multiple executions from \emph{different} programs, whereas HSL focuses on \emph{hyperproperties}, \ie properties relating multiple executions from the \emph{same} program.
Moreover, while HSL is able to express and reason about $\forall^*$-properties, it does not formally subsume RSL or LGTM, because the former enforces the same termination behavior for both executions (either both terminate, or both diverge) and the latter enforces termination for all executions, while HSL does not allow expressing hyperproperties of non-terminating executions.
Unlike HSL, which tracks a \emph{set} of reachable states,
these logics track $k$-tuples of executions, and thus defining their separating conjunction (in a pointwise manner) is straightforward,
but this limits them to reasoning about $k$-safety properties only,
whereas HSL supports arbitrary quantifier alternation.

Other specialized separation logics have been developed to reason about specific hyperproperties,
for example SecCSL~\cite{ernstSecCSLSecurityConcurrent2019} and CommCSL~\cite{eilersCommCSLProvingInformation2023}
target non-interference (\ie $\forall \forall$) in concurrent programs, while
Simuliris~\cite{gaherSimulirisSeparationLogic2022} and ReLoC~\cite{fruminReLoCMechanisedRelational2018}
target (contextual) refinement (\ie $\forall \exists$).

\section{Conclusion and Future Work}
\label{sec:conclusion}
We presented \emph{Hyper Separation Logic} (HSL), the first separation logic that supports reasoning about hyperproperties with arbitrary quantifier alternation, covering both $\forall^*\exists^*$ properties (e.g., generalized non-interference) and $\exists^*\forall^*$ properties (e.g., existence of a maximum).
At its core is a \emph{hyper separating conjunction} that composes hyper-assertions by preserving their existential lower bounds and universal upper bounds on states.
This connective enables sound, generic rules for heap-manipulating commands and a generalized frame rule for local reasoning, whose soundness follows from embedding all $\forall$-frames into triple validity.
Reachability reasoning is supported by labeled states, including the novel label $\uk$.

Future work includes extending HSL to concurrency, to termination reasoning (both proving and disproving), and to relational properties across multiple programs.

\begin{acks}
	We are deeply grateful to Tinko Tinchev for numerous insightful discussions and significant contributions to the conceptual development and technical proofs presented in this work.
	This research was partially funded by the Ministry of Education and Science of Bulgaria (support for INSAIT, part of the Bulgarian National Roadmap for Research Infrastructure).
\end{acks}

\section*{Data Availability Statement}
All technical results presented in this paper have been formalized and proven in Isabelle/HOL, and our formalization is publicly available \citep{HyperSeparationLogic_artifact}.

\bibliography{references}

\clearpage
\appendix

\section{Technical Definitions}
\label{sec:appendix}
\begin{figure}
	\begingroup
	\newcommand{\scaleabortfig}{0.75}
	\scalebox{\scaleabortfig}{$\prftree[r]{}{}{\la x\coloneqq e,\la s,h\ra\ra\to\la\resv{skip},\la \update{s}{x}{e(s)},h\ra\ra}$}\hfill 
	\scalebox{\scaleabortfig}{$\prftree[r]{}{}{\la x\coloneqq\mathrm{nonDet()},\la s,h\ra\ra\to\la\resv{skip},\la \update{s}{x}{v},h\ra\ra}$}\hfill 
	\scalebox{\scaleabortfig}{$\prftree[r]{}{\lnot b(s)}{\la\Assert{b},\la s,h\ra\ra\to\cEr}$} 
	\\[\bigskipamount]
	\scalebox{\scaleabortfig}{$\prftree[r]{}{l\not=0}{l\notin\mathrm{Dom}(h)}{\la x\coloneqq\mathrm{alloc()},\la s,h\ra\ra\to\la\resv{skip},\la \update{s}{x}{l},\update{h}{l}{v}\ra\ra}$}\hfill 
	\scalebox{\scaleabortfig}{$\prftree[r]{}{b(s)}{\la\resv{assume}\ b,\la s,h\ra\ra\to\la\resv{skip},\la s,h\ra\ra}$}\hfill 
	\scalebox{\scaleabortfig}{$\prftree[r]{}{b(s)}{\la\Assert{b},\la s,h\ra\ra\to\la\resv{skip},\la s,h\ra\ra}$} 
	\\[\bigskipamount]
	\scalebox{\scaleabortfig}{$\prftree[r]{}{s(x)\not=0}{h(s(x))=v}{\la[x]\coloneqq e,\la s,h\ra\ra\to\la\resv{skip},\la s,\update{h}{s(x)}{e(s)}\ra\ra}$} \hfill
	\scalebox{\scaleabortfig}{$\prftree[r]{}{s(x)=0 \vee h(s(x))=\bot}{\la[x]\coloneqq e,\la s,h\ra\ra\to\cEr}$}\hfill 
	\scalebox{\scaleabortfig}{$\prftree[r]{}{s(x)\not=0}{s(x)\notin\mathrm{Dom}(h)}{\la[x]\coloneqq e,\la s,h\ra\ra\to\cDv}$} 
	\\[\bigskipamount]
	\scalebox{\scaleabortfig}{$\prftree[r]{}{s(x)\not=0}{h(s(x))=v}{\la y\coloneqq[x],\la s,h\ra\ra\to\la\resv{skip},\la \update{s}{y}{v},h\ra\ra}$}\hfill 
	\scalebox{\scaleabortfig}{$\prftree[r]{}{s(x)=0 \vee h(s(x))=\bot}{\la y\coloneqq[x],\la s,h\ra\ra\to\cEr}$}\hfill 
	\scalebox{\scaleabortfig}{$\prftree[r]{}{s(x)\not=0}{s(x)\notin\mathrm{Dom}(h)}{\la y\coloneqq[x],\la s,h\ra\ra\to\cDv}$} 
	\\[\bigskipamount]
	\scalebox{\scaleabortfig}{$\prftree[r]{}{s(x)\not=0}{h(s(x))=v}{\la\mathrm{free}(x),\la s,h\ra\ra\to\la\resv{skip},\la s,\update{h}{s(x)}{\bot}\ra\ra}$}\hfill 
	\scalebox{\scaleabortfig}{$\prftree[r]{}{s(x)=0 \vee h(s(x))=\bot}{\la \mathrm{free}(x),\la s,h\ra\ra\to\cEr}$}\hfill 
	\scalebox{\scaleabortfig}{$\prftree[r]{}{s(x)\not=0}{s(x)\notin\mathrm{Dom}(h)}{\la \mathrm{free}(x),\la s,h\ra\ra\to\cDv}$}\hfill 
	\\[\bigskipamount]
	\scalebox{\scaleabortfig}{$\prftree[r]{}{}{\la\resv{skip};C_2,\sigma\ra\to\la C_2,\sigma\ra}$}\hfill 
	\scalebox{\scaleabortfig}{$\prftree[r]{}{\la C_1,\sigma\ra\to\la C_1',\sigma'\ra}{\la C_1;C_2,\sigma\ra\to\la C_1';C_2,\sigma'\ra}$}\hfill 
	\scalebox{\scaleabortfig}{$\prftree[r]{}{\la C_1,\sigma\ra\to\cEr}{\la C_1;C_2,\sigma\ra\to\cEr}$}\hfill 
	\scalebox{\scaleabortfig}{$\prftree[r]{}{\la C_1,\sigma\ra\to\cDv}{\la C_1;C_2,\sigma\ra\to\cDv}$} 
	\\[\bigskipamount]
	\scalebox{\scaleabortfig}{$\prftree[r]{}{}{\la C_1+C_2,\sigma\ra\to\la C_1,\sigma\ra}$}\hfill
	\scalebox{\scaleabortfig}{$\prftree[r]{}{}{\la C_1+C_2,\sigma\ra\to\la C_2,\sigma\ra}$}\hfill
	\scalebox{\scaleabortfig}{$\prftree[r]{}{}{\la C^*,\sigma\ra\to\la C;C^*,\sigma\ra}$}\hfill
	\scalebox{\scaleabortfig}{$\prftree[r]{}{}{\la C^*,\sigma\ra\to\la\resv{skip},\sigma\ra}$}
	\caption{Small-step semantics of program commands, where $v\in\mathbb{N}$ and $\update{\_}{x}{v}$ denotes a function update.} \label{fig:small-step-semantics-def}
	\endgroup
\end{figure}
The small-step semantics for the program commands introduced in \secref{subsec:formal:semantics} are shown in \figref{fig:small-step-semantics-def} and $\md[C]$ is defined in \figref{fig:md-definition}.

We now formalize the interpretation of the syntax from \secref{subsec:formal:syntax}, which, as hinted there, is defined in two stages.
The first stage operates on \emph{intermediate states} $\hat{\omega}$, consisting of a pair of a \emph{scaffold store} \Deleatur\ (a total function from scaffold variables in $\GVars$ to $\mathbb{N}$)
and an extended state $\omega$, allowing us to bind pointed-to values and reference them throughout this stage.
The second stage simply projects \emph{intermediate assertions} $\hat{p}$---sets of intermediate states---into standard assertions $p$, i.e., sets of states.
We begin by defining stage one:
\[\begin{aligned}
	\hat{S},I,s_0 &\models_0 \top \triangleq \mathrm{true} \\
	\hat{S},I,s_0 &\models_0 \bot \triangleq \mathrm{false} \\
	\hat{S},I,s_0 &\models_0 \sigma(\hat{p}) \triangleq \sigma\in\mathrm{Dom}(I)\land\la\text{\Deleatur},\la\Lambda,\la s\triangleleft s_0,h\ra\el\ra\ra\in\hat{p}\text{, where }I(\sigma)=\la\text{\Deleatur},\la\Lambda,\la s,h\ra\el\ra\ra \\
	\hat{S},I,s_0 &\models_0 \en{n}P \triangleq \exists v\ldotp\hat{S},I,s_0[n\mapsto v]\models_0 P \\
	\hat{S},I,s_0 &\models_0 \fn{n}P \triangleq \forall v\ldotp\hat{S},I,s_0[n\mapsto v]\models_0 P \\
	\hat{S},I,s_0 &\models_0 \es{\sigma}P \triangleq \exists\hat{\omega}\in\hat{S}\ldotp\hat{S},I[\sigma\mapsto\hat{\omega}],s_0\models_0 P \\
	\hat{S},I,s_0 &\models_0 \fs{\sigma}P \triangleq \forall\hat{\omega}\in\hat{S}\ldotp\hat{S},I[\sigma\mapsto\hat{\omega}],s_0\models_0 P \\
	\hat{S},I,s_0 &\models_0 P\land Q \triangleq \hat{S},I,s_0\models_0 P\land\hat{S},I,s_0\models_0 Q \\
	\hat{S},I,s_0 &\models_0 P\lor Q \triangleq \hat{S},I,s_0\models_0 P\lor\hat{S},I,s_0\models_0 Q \\
	\hat{S},I,s_0 &\models_0 P\otimes Q \triangleq I=\emptyset\land\exists\hat{S_P},\hat{S_Q}\ldotp \hat{S}=\hat{S_P}\cup\hat{S_Q}\land\hat{S_P},I,s_0\models_0 P\land\hat{S_Q},I,s_0\models_0 Q \\
	\hat{S},I,s_0 &\models_0 P\star Q \triangleq \exists\hat{S_P},\hat{S_Q},I_P,I_Q\ldotp\hat{S}\subseteq\hat{S_P}*\hat{S_Q}\land\plw[\hat{S}][\hat{S_P}][\hat{S_Q}]\land\prw[\hat{S}][\hat{S_P}][\hat{S_Q}]\\ 
	&\qquad\qquad\, \land \mathrm{Dom}(I)=\mathrm{Dom}(I_P)=\mathrm{Dom}(I_Q)\land\mathrm{Ran}(I_P)\subseteq\hat{S_P}\land\mathrm{Ran}(I_Q)\subseteq\hat{S_Q}\\
	&\qquad\qquad\, \land (\forall\sigma\in\mathrm{Dom}(I)\ldotp\addStates[I(\sigma)][I_P(\sigma)][I_Q(\sigma)])\land \hat{S_P},I_P,s_0\models_0 P\land \hat{S_Q},I_Q,s_0\models_0 Q
\end{aligned}\]
where
$\hat{S}$ is an intermediate assertion, $I$ is a partial function from $\SVars$ to the set of intermediate assertions, $s_0$ is a partial function from $\PVars$ to $\mathbb{N}$ and
\[(f\triangleleft f_0)(x)\triangleq
\begin{cases}
	f_0(x), & x\in\mathrm{Dom}(f_0) \\
	f(x), & x\notin\mathrm{Dom}(f_0)
\end{cases}\] is function overriding: it coincides with $f$ except on the domain of $f_0$, \mbox{where it takes the values of $f_0$}.

\begin{figure}
	\centering
	\begin{equation*}
		\begin{aligned}
			&\text{md}(\Skip) = \emptyset && \text{md}(\Assign{x}{e}) = \{x\} && \text{md}(\Havoc{x}) = \{x\} \\
			&\text{md}(\Assume{b}) = \emptyset && \text{md}(\Assert{b}) = \emptyset && \text{md}(\Alloc{x}) = \{x\} \\
			&\text{md}(\Write{x}{e}) = \emptyset && \text{md}(\Read{y}{x}) = \{y\} && \text{md}(\Free{x}) = \emptyset \\
			&\text{md}(C_1; C_2) = \text{md}(C_1) \cup \text{md}(C_2) && \text{md}(C_1+C_2) = \text{md}(C_1) \cup \text{md}(C_2) && \text{md}(C^*) = \text{md}(C)
		\end{aligned}
	\end{equation*}
	\caption{Definition of modified variables ($\text{md}$) for program commands.}
	\label{fig:md-definition}
\end{figure}

The partial functions $I$ and $s_0$ record the bindings introduced during interpretation.
$I$ maps each quantified state variable $\sigma$ to the intermediate state assigned to it at the point of quantification, while $s_0$ maps program variables $x$ to their assigned values.
That is, when $I(\sigma)$ (resp. $s_0(x)$) is defined, it indicates that the state (resp. program) variable has been previously quantified and is associated with that particular intermediate state (resp. value).
When it is undefined, the variable lies outside the scope of any quantifier.

The two most noteworthy cases are the base case, $\sigma(\hat{p})$, and the hyper separating conjunction.
The base case $\sigma(\hat{p})$ is satisfied by those $\hat{S}$, $I$ and $s_0$ for which $\sigma$ has already been instantiated via an intermediate state $I(\sigma)=\la\text{\Deleatur},\la\Lambda,\la s,h\ra\el\ra\ra$,
which, once the quantified program variables $s_0$ are accounted for (i.e., used to override $s$), belongs to $\hat{p}$.
For example, 
\scriptsize
\begin{align*}
	\hat{S},\emptyset,\emptyset\models_0\en{n}\fs{\sigma}\sigma(\ok:x=n) & \Longleftrightarrow\exists n_0\in\mathbb{N}\ldotp\hat{S},\emptyset,[n\mapsto n_0]\models_0\fs{\sigma}\sigma(\ok:x=n) \\
																	   & \Longleftrightarrow\exists n_0\in\mathbb{N}\ldotp\forall\la\text{\Deleatur},\la\Lambda,\la s,h\ra\el\ra\ra\in \hat{S}\ldotp \hat{S},[\sigma\mapsto\la\text{\Deleatur},\la\Lambda,\la s,h\ra\el\ra\ra],[n\mapsto n_0]\models_0\sigma(\ok:x=n) \\
																	   & \Longleftrightarrow \exists n_0\in\mathbb{N}\ldotp\forall\la\text{\Deleatur},\la\Lambda,\la s,h\ra\el\ra\ra\in \hat{S}\ldotp\la\text{\Deleatur},\la\Lambda,\la s[n\mapsto n_0],h\ra\el\ra\in\llbracket\ok:x=n\rrbracket \\
																	   & \Longleftrightarrow \exists n_0\in\mathbb{N}\ldotp\forall\la\text{\Deleatur},\la\Lambda,\la s,h\ra\el\ra\ra\in \hat{S}\ldotp \epsilon=\ok\land s(x)=n_0
\end{align*}
\normalsize
Starting with no quantified variables, $I=s=\emptyset$, we first record in $s_0$ the natural number $n_0$ obtained when quantifying the program variable $n$.
After that, we record in $I$ the intermediate state $\hat{\omega}=\la\text{\Deleatur},\la\Lambda,\la s,h\ra\el\ra\ra\in\hat{S}$ obtained when quantifying the state variable $\sigma$ and finally 
we check whether the overriden $\hat{\omega}$ satisfies $\ok:x=n$\footnote{Technically, we use semantic intermediate assertions, and when we write $\ok:x=n$, we actually mean the set of all intermediate states that satisfy it, denoted $\llbracket\ok:x=n\rrbracket$. For readability, we write it in the simpler syntactic form.}.

In our mechanization, quantification over program variables is implemented using De Bruijn indices and hence no overriding, $\triangleleft$, is occurring.
Instead, the type of $\hat{p}$ is more involved and takes both $s$ and the "De Bruijn store" $\Delta$ (corresponding to $s_0$):
\[\hat{S},I,\Delta\models_0 \sigma(\hat{p}) \triangleq \sigma\in\mathrm{Dom}(I)\land\langle I(\sigma),\Delta\rangle\in\hat{p}\]

We now turn to the hyper separating conjunction.
As with the semantic $\star$, the syntactic version splits $\hat{S}$ into $\hat{S_P}\models_0 P$ and $\hat{S_Q}\models_0 Q$ satisfying
\[\hat{S}\subseteq\hat{S_P}*\hat{S_Q}\land\plw[\hat{S}][\hat{S_P}][\hat{S_Q}]\land\prw[\hat{S}][\hat{S_P}][\hat{S_Q}]\]
In addition, however, we must ensure that the quantified states are split appropriately.
Specifically, we require two partial functions $I_P$ and $I_Q$ with the same domains as $I$, such that for all $\sigma\in\mathrm{Dom}(I)$,
\[\addStates[I(\sigma)][I_P(\sigma)][I_Q(\sigma)]\]
Finally, the set $\hat{S}$ and the partial function $I$ splits must be coherent with each other, meaning
\[\mathrm{Ran}(I_P)\subseteq\hat{S_P}\quad\text{and}\quad\mathrm{Ran}(I_Q)\subseteq\hat{S_Q}\]

We next examine the role of scaffold variables.
Binding the value stored at a pointer to a term that is visible to the surrounding context is essential. Otherwise, we end up with an expression such as $x\mapsto\sth*p\equiv(\en{n}x\mapsto n)*p$,
which asserts the existence of some memory cell at address $x$ but does not expose its contents to $p$.
At first glance, one might try to avoid this by simply writing $x\mapsto n*p$, so that $n$ is visible to $p$.
The problem, however, is that this treats $n$ as a free variable (with fixed value), rather than as the value stored in the heap at address $x$.
The usual solution is to quantify $n$ explicitly, yielding $\en{n}x\mapsto n*p$, which makes the stored value visible to $p$,
with $n$ remaining free and serving simply as a name for the stored value rather than a pre-existing variable.

Similarly, in the hyper setting, $\big(\fs{\sigma}\sigma(x\mapsto\sth)\big)\star P\equiv\big(\fs{\sigma}\sigma(\en{n}x\mapsto n)\big)\star P$
asserts that each state has a memory allocated at $x$, but the value stored there is not made visible to $P$.
Following the unary approach and removing the quantifier, we obtain $\big(\fs{\sigma}\sigma(x\mapsto n)\big)\star P$, which makes the memory values visible in $P$.
However, unlike in the unary setting, globally quantifying $n$ does not simply free it as intended;
it also forces all memory cells at $x$ across states to share the same value: $\en{n}\big(\fs{\sigma}\sigma(x\mapsto n)\big)\star P$.
Scaffold variables address this issue by making the memory values accessible in $P$ while effectively acting as globally quantified names that are instantiated independently in each state:
$\big(\fs{\sigma}\sigma(x\mapsto\delta_x)\big)\star P$.

The second stage of the interpretation implements precisely this notion of implicit, per-state global quantification.
However, rather than allowing all possible values as in the standard existential case, we instead eliminate any trace of the scaffold variables entirely---essentially projecting them away:
\[S\models P\defiff\exists\hat{S}\ldotp S=\{\omega\mid\la\text{\Deleatur},\omega\ra\in\hat{S}\}\land\hat{S},\emptyset,\emptyset\models_0 P\]

Entailment, $P\models Q$, is defined as $\forall\hat{S},I,s_0\ldotp\hat{S},I,s_0\models P\land\mathrm{Ran}(I)\subseteq\hat{S}\Longrightarrow\hat{S},I,s_0\models Q$
and equivalence, $P\equiv Q$, is defined as $P\models Q\land Q\models P$.
We denote with $\llbracket P\rrbracket$\footnote{As done repeatedly in the main text, brackets are sometimes omitted when context suffices.} the hyper-assertion $\{S\mid S\models P\}$.

We have established in Isabelle/HOL the following properties:
\begin{lemma}\label{lem:syntax-properties}
	The following hold
	\begin{enumerate}
		\setcounter{enumi}{-1} 
		\item $\sigma(\hat{p})\star\sigma(\hat{q})\equiv\sigma(\hat{p}*\hat{q})$
		\item $(\fs{\sigma}P)\star(\fs{\sigma}Q)\models\fs{\sigma}P\star Q$
		\item $(\es{\sigma}P)\star(\fs{\sigma}Q)\models\es{\sigma}P\star Q$
		\item $\llbracket(\es{\sigma}\es{\sigma'}P)\star(\fs{\sigma}\fs{\sigma'}Q)\rrbracket\subseteq\llbracket\es{\sigma}\es{\sigma'}P\star Q\rrbracket$
		\item $\llbracket(\es{\sigma}P\star Q)\otimes\top\rrbracket\subseteq\llbracket((\es{\sigma}P)\star(\fs{\sigma}Q))\otimes\top\rrbracket$, no $\exists\la\bm\cdot\ra$ in $P,Q$
		\item $\llbracket P\rrbracket=\llbracket P\otimes\top\rrbracket$, no $\fs{\bm\cdot}$ in $P$
		\item $\llbracket P\star Q\rrbracket\subseteq\llbracket P\rrbracket\star\llbracket Q\rrbracket$, $P,Q$ - closed (no free state variables)
		\item\label{testtest} $\llbracket P\rrbracket\star\llbracket Q\rrbracket\subseteq\llbracket P\star Q\rrbracket$, no intersecting scaffold variables in $P$ and $Q$
	\end{enumerate}
\end{lemma}

Note that we distinguish between $P\models Q$ and $\llbracket P\rrbracket\subseteq\llbracket Q\rrbracket$:
the former applies to arbitrary formulae, while the latter is only meaningful for closed formulae, as it is trivially satisfied otherwise.
In particular, $P\models Q$ implies $\llbracket P\rrbracket\subseteq\llbracket Q\rrbracket$.

On a separate note, since scaffold variable work as a global (per-state) quantification, property \eqref{testtest} cannot generally be expected to hold:
$\llbracket\fs{\sigma}\sigma(\delta=5)\rrbracket\star\llbracket\fs{\sigma}\sigma(\delta=6)\rrbracket=\UNIV\okl\star\UNIV\okl=\UNIV\okl=\Pow[\{\la\Lambda,\sigma\el\ra\mid\epsilon=\ok\}]$,
whereas $\llbracket\big(\fs{\sigma}\sigma(\delta=5)\big)\star\big(\fs{\sigma}\sigma(\delta=6)\big)\rrbracket=\emptyset$.
This is not problematic, as one can simply rename any intersecting scaffold variables.

\section{Rules and Proof Sketches}
\label{sec:appendix-soundness}
\begin{figure}
	\footnotesize
	\begin{mathpar}
		\newrule{TrivPost}{ax:triv-post}{\hoare{P}{C}{\top}}
		\and
		\newrule{IdxUnion}{rule:idx-union}[\forall i\in I\ldotp\hoare{P_i}{C}{Q_i}]{\hoare{\mcup_{i\in I}P_i}{C}{\mcup_{i\in I}Q_i}}
		\and
		\newrule{IdxJoin}{rule:idx-join}[\forall i\in I\ldotp\hoare{P_i}{C}{Q_i}]{\hoare{\motimes_{i\in I}P_i}{C}{\motimes_{i\in I}Q_i}}
		\and
		\newrule{Iter}{rule:iter}[\forall n\in\mathbb{N}\ldotp\hoare{I_n}{C}{I_{n+1}}]{\hoare{I_0}{C^*}{\motimes_{n\in\mathbb{N}}I_n}}
		\and
		\newrule{Skip}{ax:skip}{\hoare{P}{\Skip}{P}}
		\and
		\newrule{If}{rule:if}[\hoare{P}{C_1}{Q_1}][\hoare{P}{C_2}{Q_2}]{\hoare{P}{C_1+C_2}{Q_1\otimes Q_2}}
		\and
		\newrule{Cons+}{rule:cons-plus}[\hoare{P}{C}{Q}][P'\subseteq P][\forall f\in\F[\md[C]]\ldotp\R[Q\star\Pow[f]]\subseteq\R[Q'\star\Pow[f]]]{\hoare{P'}{C}{Q'}}
		\and
		\newrule{Assign}{ax:assign}{\hoare{\{S\mid\{\la\Lambda,\la s[x\mapsto e(s)], h\ra\okl\ra\mid\la\Lambda,\la s,h\ra\okl\ra\in S\}\cup\{\la\Lambda,\sigma\el\ra\in S\mid\epsilon\not=\ok\}\in Q\}}{\Assign{x}{e}}{Q}}
		\and
		\newrule{Havoc}{ax:havoc}{\hoare{\{S\mid\{\la\Lambda,\la s[x\mapsto n], h\ra\okl\ra\mid\la\Lambda,\la s,h\ra\okl\ra\in S\}\cup\{\la\Lambda,\sigma\el\ra\in S\mid\epsilon\not=\ok\}\in Q\}}{\Havoc{x}}{Q}}
		\and
		\newrule{Assume}{ax:assume}{\hoare{\{S\mid\{\la\Lambda,\la s,h\ra\el\ra\in S\mid \epsilon=\ok\Rightarrow b(s)\}\in Q\}}{\Assume{b}}{Q}}
		\and
		\newrule{Assert}{ax:assert}{\hoare{\{S\mid\{\la\Lambda,\la s,h\ra\el\ra\in S\mid \epsilon=\ok\Rightarrow b(s)\}\cup\{\la\Lambda,\la s,h\ra\erl\ra\mid\la\Lambda,\la s,h\ra\okl\ra\in S\land\lnot b(s)\}\in Q\}}{\Assert{b}}{Q}}
		\and
		\newrule{WriteNull}{ax:write-null}{\hoare{(\fs{\sigma}\sigma(\ok:x=0))\star P}{\Write{x}{e}}{(\fs{\sigma}\sigma(\er:x=0))\star P}}
		\and
		\newrule{WriteFrd}{ax:write-frd}{\hoare{(\fs{\sigma}\sigma(\ok:x\mapsto\bot))\star P}{\Write{x}{e}}{(\fs{\sigma}\sigma(\er:x\mapsto\bot))\star P}}
		\and
		\newrule{WriteEr}{ax:write-er}{\hoare{(\fs{\sigma}\sigma(\ok:x=0\lor x\mapsto\bot))\star P}{\Write{x}{e}}{(\fs{\sigma}\sigma(\er:x=0\lor x\mapsto\bot))\star P}}
		\and
	\end{mathpar}
	\caption{Additional rules of Hyper Separation Logic.}
	\label{fig:additional-rules}
\end{figure}

In this appendix, we present additional selected rules of Hyper Separation Logic (\figref{fig:additional-rules}), illustrating cases not covered in the main text.
Analogues of \ruleref{ax:write-null}, \ruleref{ax:write-frd}, and \ruleref{ax:write-er} have been similarly established as sound for the read and free operations.

As noted in the main text, we have formally verified \emph{semantic} soundness\footnote{That is, implications are proven between valid triples $\hoare{\bm\cdot}{\!\!\!\bm\cdot\!\!\!}{\bm\cdot}\Rightarrow\hoare{\bm\cdot}{\!\!\!\bm\cdot\!\!\!}{\bm\cdot}$, rather than derived ones $\shoare{\bm\cdot}{\!\!\!\bm\cdot\!\!\!}{\bm\cdot}\Rightarrow\shoare{\bm\cdot}{\!\!\!\bm\cdot\!\!\!}{\bm\cdot}$.} in Isabelle/HOL.
While the rules in the paper are mostly presented syntactically, our formal development establishes stronger, semantically grounded\footnote{That is, the triples $\simpleHoare{P}{C}{Q}$ are formed using semantic hyper-assertions $P$ and $Q$ rather than syntactic ones.} versions.
For instance, the semantic counterpart of the side condition $\fv[F]\cap\md[C]=\emptyset$ used in the \ruleref{rule:frame} rule is given by
\[\forall S,S'\ldotp S\approx_{\md[C]}S'\Longrightarrow(S\in F\Longleftrightarrow S'\in F)\]
where
$S\approx_{\mathrm{vars}} S'$ holds whenever
\begin{enumerate}
	\item $\forall\la\Lambda,\la s,h\ra\el\ra\in S\ldotp\exists s'\ldotp(\forall x\notin\mathrm{vars}\ldotp s'(x)=s(x))\land\la\Lambda,\la s',h\ra\el\ra\in S'$; and\hfill \textit{//denoted $S\precsim_{\mathrm{vars}}S'$}
	\item $\forall\la\Lambda,\la s',h\ra\el\ra\in S'\ldotp\exists s\ldotp(\forall x\notin\mathrm{vars}\ldotp s(x)=s'(x))\land\la\Lambda,\la s,h\ra\el\ra\in S$.\hfill \textit{//denoted $S'\precsim_{\mathrm{vars}}S$}
\end{enumerate}
Below, we give two proof sketches---one for \ruleref{rule:frame} and one for \ruleref{ax:read}---of the more general semantic variants, which in turn entail the syntactic rules presented in \figref{fig:selected-rules}.

We begin by demonstrating the soundness of the \ruleref{rule:frame} rule, which forms a cornerstone of separation logic.
While this rule is typically among the most challenging, embedding all admissible $\forall$-frames into the definition of validity significantly eases its proof.
However, even though, as discussed in \secref{subsec:key:embedding-fr}, $\forall^+$-frames can be trivially expressed via an infinite disjunction of $\forall$-frames,
it is not immediately obvious that these $\forall$-frames are actually admissible. We proceed to establish that this is indeed the case%
\footnote{For readability, we write $\fv[F]\cap\mathrm{vars}=\emptyset$ and $\fv[f]\cap\mathrm{vars}=\emptyset$, though we formally mean their semantic counterparts.}:
\begin{lemma}\label{lem:split-downward-closed-to-union-of-powersets}
	Let $F$ be downward closed and be such that $\fv[F]\cap\mathrm{vars}=\emptyset$.
	Then there exists $F_{\mathrm{max}}\subseteq F$ such that $F=\bigcup_{f\in F_{\mathrm{max}}}\mathcal{P}(f)$ and $\forall f\in F_{\mathrm{max}}\ldotp\fv[f]\cap\mathrm{vars}=\emptyset$.
\end{lemma}
\begin{proof}[Proof sketch]
	Let $F_{\mathrm{max}}\leftrightharpoons\{S'\in F\mid\forall S\approx_{\mathrm{vars}}S'\ldotp S\subseteq S'\}\subseteq F$.
	Using the fact that $\approx_{\mathrm{vars}}$ is an equivalence relation, the fact that for each $S$, the set $\{S'\mid S'\approx_{\mathrm{vars}}S\}$ contains the greatest element w.r.t. $\subseteq$
	and that $\fv[F]\cap\mathrm{vars}=\emptyset$, we obtain that $F\subseteq\bigcup_{f\in F_{\mathrm{max}}}\mathcal{P}(f)$.
	The converse, $\bigcup_{f\in F_{\mathrm{max}}}\mathcal{P}(f)\subseteq F$, follows directly from the fact that $F$ is downward closed.
	Lastly, since for each $S$, the greatest element of $\{S'\mid S'\approx_{\mathrm{vars}}S\}$, $S_{max}$, is easily verified to satisfy $\fv[S_{max}]\cap\mathrm{vars}=\emptyset$, we conclude that
	$\forall f\in F_{\mathrm{max}}\ldotp\fv[f]\cap\mathrm{vars}=\emptyset$.
\end{proof}

With this key lemma, required for the soundness of the \ruleref{rule:frame}, now established, we proceed to sketch the proofs of the two rule:
\begin{proof}[Proof sketch (\ruleref{rule:frame})]
	Using Lemma \ref{lem:split-downward-closed-to-union-of-powersets}, let $F_{\mathrm{max}}\subseteq F$ be such that $F=\bigcup_{f\in F_{\mathrm{max}}}\mathcal{P}(f)$ and $\forall f\in F_{\mathrm{max}}\ldotp\fv[f]\cap\mathrm{vars}=\emptyset$.
	Now, since $F$ is $\ok$-only, and hence $\forall f\in F_{\mathrm{max}}\ldotp f$ is $\ok$-only, we obtain that $\forall f\in F_{\mathrm{max}}\ldotp f\in\F[\md[C]]$.
	Therefore \mbox{$\forall f\in F_{\mathrm{max}}\ldotp\hoare{P\star\mathcal{P}(f)}{C}{Q\star\mathcal{P}(f)}$} by the definition of validity 
	and the fact that $\mathcal{F}$ is closed under $*$, together with associativity of $\star$.
	By \ruleref{rule:idx-union} we obtain that $\hoare{\bigcup_{f\in F_{\mathrm{max}}}P\star\mathcal{P}(f)}{C}{\bigcup_{f\in F_{\mathrm{max}}}Q\star\mathcal{P}(f)}$.
	Now, since $\star$ distributes over (infinite) union, is commutative and the fact that $F=\bigcup_{f\in F_{\mathrm{max}}}\mathcal{P}(f)$, we conclude that $\hoare{P\star F}{C}{Q\star F}$.
\end{proof}
\begin{proof}[Proof sketch (\ruleref{ax:read})]
	Let $S\models(\fs{\sigma}\sigma(\ok:x\mapsto e))\star P,\baseok[P]$ and $y\notin\fv[P]\cup\const{pvars}(e)\cup\{x\}$.
	For the purposes of this proof sketch, we can safely ignore the embedded frame rule as the conjunct $\fs{\sigma}\sigma(\ok:x\mapsto e)$ already supplies the necessary resources for $\Read{x}{y}$;
	in the full proof this is handled more rigorously but is immaterial to the high-level argument.
	Consider the partitioning $S=S\okl\cup S\erl\cup S\ukl$, with $S\el$ defined as the subset of states in $S$ whose label is exactly $\epsilon$.
	Now, since $\er$ and $\uk$ states don't change under $\Sem$, we've that \[\Sem[\Read{x}{y}][S]=\Sem[\Read{x}{y}][S\okl]\cup S\erl\cup S\ukl\]
	We claim that $\Sem[\Read{x}{y}][S\okl]\cup S\erl\cup S\ukl\models\R[(\fs{\sigma}\sigma(\ok:x\mapsto e\land y=e))\star P]$.
	Indeed, consider $S'\leftrightharpoons\Sem[\Read{x}{y}][S\okl]\cup S'\ukl$, where $S'\ukl$ is obtained from $S\erl\cup S\ukl$ by updating the value of $y$ to $e$ and setting all labels to $\uk$.
	It is clear that \[\Res[S'][\Sem[\Read{x}{y}][S\okl]\cup S\erl\cup S\ukl]\] and all that remains is to show that $S'\models(\fs{\sigma}\sigma(\ok:x\mapsto e\land y=e))\star P$.
	In order to do that, recall that $S\models(\fs{\sigma}\sigma(\ok:x\mapsto e))\star P$.
	Now, consider the auxiliary set of states $S_{aux}$, obtained by replacing all $\er$ labels in $S$ to $\uk$ ones. Note that, unlike $S'\ukl$, we haven't (yet) change the values of $y$.
	Since we've that $\baseok[P]$, we have that $S_{aux}\models(\fs{\sigma}\sigma(\ok:x\mapsto e))\star P$.
	Finally, note that $S'$ is exactly $S_{aux}$ with $y$ values changed to $e$ and since $y\notin\const{pvars}(e)\cup\{x\}$, it follows that $x\mapsto e\land y=e$ holds for all states of $S'$.
	Therefore, since $y\notin\fv[P]$, we conclude that \[S'\models(\fs{\sigma}\sigma(\ok:x\mapsto e\land y=e))\star P\]
	The formal proof considers two certificates of $S\models(\fs{\sigma}\sigma(\ok:x\mapsto e))\star P$, namely $S_x\models\fs{\sigma}\sigma(\ok:x\mapsto e)$ and $S_P\models P$ and performs analogous steps as the outlined in the sketch.
	The modified $S_x$ and $S_P$ are then used as the certificates for $S'\models(\fs{\sigma}\sigma(\ok:x\mapsto e\land y=e))\star P$.
\end{proof}

We conclude this appendix with a discussion about the two consequence rules: \ruleref{rule:cons} and \ruleref{rule:cons-plus}.
\ruleref{rule:cons} is standard, while \ruleref{rule:cons-plus} is slightly more involved, particularly its third assumption. This assumption captures the weakest form of entailment required for a sound consequence rule.
The subtlety is that this entailment depends on the program variables modified, $\md[C]$: the more variables $C$ modifies, the weaker the required entailment%
\footnote{The reason for this is that $\F$ is reverse monotonic, i.e., $X\subseteq Y\Longrightarrow\F[X]\supseteq\F[Y]$.}.
For instance, $\fs{\sigma}\sigma(\ok:x=5)$ entails $\fs{\sigma}\sigma(\uk:x=6)$ if $x \in \md[C]$, but not otherwise.
\ruleref{rule:cons-plus} is strictly stronger than \ruleref{rule:cons}: $\fs{\sigma}\sigma(\ok:x=5)$ entails $\fs{\sigma}\sigma(\uk:x=5)$ independently of $C$, yet it is not a subset.

This stronger rule is mainly of interest for the development of the logic itself rather than for end users, but we include it for completeness.
It is worth noting that this is not the strongest possible consequence rule. To illustrate that, consider the following, formally proven in Isabelle/HOL, equivalence
$$
\hoare{P}{C}{Q}
\Longleftrightarrow
\forall f\in\F[\md[C]]\ldotp
\forall S \in \R[P \star\Pow[f]] \ldotp
\Sem[C][S] \in \R[Q \star\Pow[f]]
$$
The reverse direction holds trivially, since $P \subseteq \R[P]$ for any $P$.
The forward direction, on the other hand, holds because $\R$ only overapproximates unknown states and
these unknown states already have corresponding elements in the postcondition that overapproximate them (since $\hoare{P}{C}{Q}$), and those elements in turn overapproximate the set $S$.
It is now easy to see that the same weakening can be done to the second assumption of \ruleref{rule:cons-plus} as was done to its third.
For virtually all practical purposes, aside from perhaps considerations of completeness, \ruleref{rule:cons-plus} is sufficient.

Finally, we do not address completeness, as the definition of validity admits certain ill-behaved triples.
Such triples establish validity by exploiting unknown states, yet provide no meaningful information;
for example, $\hoare{\fs{\sigma}\sigma(\ok:x=5)}{\Assign{x}{x}}{\fs{\sigma}\sigma(\uk:x=6)}$.
To properly explore completeness, we would first need to "tighten" the overapproximation performed by $\R$.
Although ill-behaved valid triples exist, \thmref{th:adequacy} demonstrates that the definition is adequate for a broad class of triples.
We call such triples ill-behaved because they violate the following natural “sanity check” for the definition of validity:
\[\llbracket C\rrbracket=\llbracket C'\rrbracket\Longrightarrow(\hoare{P}{C}{Q}\Longleftrightarrow\hoare{P}{C'}{Q})\]
For example, $\llbracket\Assign{x}{x}\rrbracket=\llbracket\Skip\rrbracket$, but $\not\hoare{\fs{\sigma}\sigma(\ok:x=5)}{\Skip}{\fs{\sigma}\sigma(\uk:x=6)}$.
This may appear problematic at first, but in practice it only arises in contrived cases that rely on unknown states and does not affect the behavior of well-formed, meaningful specifications.

\section{Examples}
\label{sec:appendix-main-examples}
In this appendix we provide additional examples that highlight the expressiveness of Hyper Separation Logic.
All examples in the appendix assume affine interpretation of $\mapsto$. Moreover, conjunction has precedence over star.

\subsection{Reachability of Errors}

\begin{figure}
	\begin{minipage}{\linewidth} 
		\small
		\begin{align*}
			\hspace{0pt}
			&\hasrt{\es{\sigma}\sigma(\ok:\top)}\\
			&x\coloneqq\mathrm{alloc}()\tag{\ruleref{ax:alloc}}\\
			&\hasrt{(\fs{\sigma}\sigma(\ok:x\mapsto\_))\star(\es{\sigma}\sigma(\ok:\top))}\\
	\models &\hasrt{\es{\sigma}\sigma(\ok:x=x*x\mapsto\_)}\\
			&\Assign{y}{x}\tag{\ruleref{ax:assign}}\\
			&\hasrt{\es{\sigma}\sigma(\ok:x=y*x\mapsto\_)}\\
	\models &\hasrt{((\fs{\sigma}\sigma(\ok:y\mapsto\_))\star(\es{\sigma}\sigma(\ok:x=y)))\otimes\top}\\
			&\mathrm{free}(y)\tag{\ruleref{ax:free},\ \ruleref{rule:join-true}}\\
			&\hasrt{((\fs{\sigma}\sigma(\ok:y\mapsto\bot))\star(\es{\sigma}\sigma(\ok:x=y)))\otimes\top}\\
	\models &\hasrt{((\fs{\sigma}\sigma(\ok:x\mapsto\bot))\star(\es{\sigma}\sigma(\ok:x=y)))\otimes\top}\\
			&\Write{x}{5}\tag{\ruleref{ax:write-frd},\ \ruleref{rule:join-true}}\\
			&\hasrt{((\fs{\sigma}\sigma(\er:x\mapsto\bot))\star(\es{\sigma}\sigma(\ok:x=y)))\otimes\top}\\
			\models &\hasrt{\es{\sigma}\sigma(\er:x\mapsto\bot)}\\
		\end{align*}
		\caption{Proof that the program in black encounters an error caused by freeing aliased memory.}
		\label{fig:example-error}
	\end{minipage}%
\end{figure}

The first example (\figref{fig:example-error}) illustrates an error-reachability scenario, an $\exists$-hyperproperty.
The program allocates memory at $x$, then creates an alias by assigning $x$ to $y$.
Freeing the memory through $y$ therefore also frees the location referenced by $x$, making the subsequent deallocation at $x$ erroneous.

\subsection{Monotonicity}

\begin{figure}
	\begin{minipage}{\linewidth}
		\footnotesize
		\begin{align*}
			\hspace{10pt}
			&\hasrt{\Mono[x][t][\delta_x]\star\Mono[y][t][\delta_y]}\\
			\models &\hasrt{(\fs{\sigma}\sigma(x\mapsto\delta_x))\!\star\!(\fs{\sigma}\sigma(y\mapsto\delta_y))\!\star\!(\fs{\sigma_1}\!\fs{\sigma_2}\!\en{n}\sigma_1(t\!=\!0\Rightarrow\delta_x\!+\!\delta_y=n)\!\land\!\sigma_2(t\!\neq\!0\Rightarrow\delta_x\!+\!\delta_y>n))}\\
			&\Read{v_x}{x}\tag{\ruleref{ax:read-scaffold}}\\
			&\hasrt{(\fs{\sigma}\sigma(x\mapsto\delta_x\land v_x=\delta_x))\star(\fs{\sigma}\sigma(y\mapsto\delta_y))\star(\fs{\sigma_1}\fs{\sigma_2}\en{n}\cdots)}\\
			\models &\hasrt{(\fs{\sigma}\sigma(y\mapsto\delta_y))\star(\fs{\sigma}\sigma(x\mapsto\delta_x\land v_x=\delta_x))\star(\fs{\sigma_1}\fs{\sigma_2}\en{n}\cdots)}\\
			&\Read{v_y}{y}\tag{\ruleref{ax:read-scaffold}}\\
			&\hasrt{(\fs{\sigma}\sigma(y\mapsto\delta_y\land v_y=\delta_y))\star(\fs{\sigma}\sigma(x\mapsto\delta_x\land v_x=\delta_x))\star\cdots}\\
			&\Alloc{z}\tag{\ruleref{ax:alloc}}\\
			&\hasrt{(\fs{\sigma}\sigma(z\mapsto\_))\star(\fs{\sigma}\sigma(y\mapsto\delta_y\land v_y=\delta_y))\star(\fs{\sigma}\sigma(x\mapsto\delta_x\land v_x=\delta_x))\star\cdots}\\
			&\Write{z}{v_x+v_y}\tag{\ruleref{ax:write}}\\
			&\hasrt{(\fs{\sigma}\sigma(z\mapsto v_x+v_y))\star(\fs{\sigma}\sigma(y\mapsto\delta_y\land v_y=\delta_y))\star(\fs{\sigma}\sigma(x\mapsto\delta_x\land v_x=\delta_x))\star\cdots}\\	
			\models &\hasrt{(\fs{\sigma}\sigma(z\mapsto v_x+v_y)\star\sigma(y\mapsto\delta_y\land v_y=\delta_y)\star\sigma(x\mapsto\delta_x\land v_x=\delta_x))\star\cdots}\\
			\models &\hasrt{(\fs{\sigma}\sigma(z\mapsto v_x+v_y*y\mapsto\delta_y\land v_y=\delta_y*x\mapsto\delta_x\land v_x=\delta_x))\star\cdots}\\
			\models &\hasrt{(\fs{\sigma}\sigma(z\mapsto\delta_x+\delta_y)\star(\fs{\sigma_1}\fs{\sigma_2}\en{n}\sigma_1(t=0\Rightarrow\delta_x+\delta_y=n)\land\sigma_2(t\neq0\Rightarrow\delta_x+\delta_y>n))}\\
			\models &\hasrt{\Mono[z][t][\delta_z]}
		\end{align*}
		\caption{Proof that the program in black satisfies monotonicity.}
		\label{fig:mono-example}
	\end{minipage}
\end{figure}

The second example (\figref{fig:mono-example}) illustrates monotonicity, a $\forall\forall$-hyperproperty.
The initial precondition asserts that both $x$ and $y$ are allocated, with the values they point to always larger when  $t\neq0$ than when $t=0$:
\[\Mono[x][t][\delta_x]\triangleq(\fs{\sigma}\sigma(x\mapsto\delta_x))\star(\fs{\sigma_1}\fs{\sigma_2}\en{n}\sigma_1(t=0\Rightarrow\delta_x=n)\land\sigma_2(t\neq0\Rightarrow\delta_x>n))\]
We then read the pointed-to values into $v_x$ and $v_y$ and allocate a new heap location $z$ to store their sum.
This demonstrates that, within our logic, monotonicity can be formally proven to be preserved under addition.

\subsection{Non-interference}

\begin{figure}
	\begin{minipage}{\linewidth}
		\footnotesize
		\begin{align*}
			\hspace{0pt}
			&\hasrt{\lowa[x][\delta_x]}\\
			&v\coloneqq[x]\tag{\ruleref{ax:read-scaffold}}\\
			&\hasrt{(\fs{\sigma}\sigma(x\mapsto\delta_x\land v=\delta_x))\star(\fs{\sigma_1}\fs{\sigma_2}\en{n}\sigma_1(\delta_x=n)\land\sigma_2(\delta_x=n))}\\
	\models &\hasrt{(\fs{\sigma}\sigma(x\mapsto\sth))\star(\fs{\sigma_1}\fs{\sigma_2}\en{n}\sigma_1(v=n)\land\sigma_2(v=n))}\\
			&\mathrm{free}(x)\tag{\ruleref{ax:free}}\\
			&\hasrt{(\fs{\sigma}\sigma(x\mapsto\bot))\star(\fs{\sigma_1}\fs{\sigma_2}\en{n}\sigma_1(v=n)\land\sigma_2(v=n))}\\
	\models &\hasrt{\fs{\sigma_1}\fs{\sigma_2}\en{n}\sigma_1(v=n)\land\sigma_2(v=n)}\\
			&y\coloneqq\mathrm{alloc}()\tag{\ruleref{ax:alloc}}\\
			&\hasrt{(\fs{\sigma}\sigma(y\mapsto\_))\star(\fs{\sigma_1}\fs{\sigma_2}\en{n}\sigma_1(v=n)\land\sigma_2(v=n))}\\
			&[y]\coloneqq2\cdot v+1\tag{\ruleref{ax:write}}\\
			&\hasrt{(\fs{\sigma}\sigma(y\mapsto2v+1))\star(\fs{\sigma_1}\fs{\sigma_2}\en{n}\sigma_1(v=n)\land\sigma_2(v=n))}\\
	\models &\hasrt{(\fs{\sigma}\sigma(y\mapsto2v+1))\star(\fs{\sigma_1}\fs{\sigma_2}\en{n}\sigma_1(2v+1=n)\land\sigma_2(2v+1=n))}\\
	\models &\hasrt{\lowa[y][\delta_y]}
		\end{align*}
		\caption{Proof that the program in black satisfies non-interference.}
		\label{fig:non-interference-example}
	\end{minipage}
\end{figure}

The third example (\figref{fig:non-interference-example}) illustrates non-interference, a $\forall\forall$-hyperproperty.
The initial precondition
\[\lowa[x][\delta_x]\triangleq(\fs{\sigma}\sigma(x\mapsto\delta_x))\star(\fs{\sigma}\fs{\sigma'}\en{n}\sigma(\delta_x=n)\land\sigma'(\delta_x=n))\]
asserts that in all starting states, $x$ is allocated with the same value. First, we read this value into $v$ and free it.
Then, we allocate a new heap location $y$ and store a deterministically modified version of the original value.
Since the original value was identical across all states, the deterministic transformation preserves equality, maintaining non-interference.

\subsection{Non-determinism}

\begin{figure}
	\begin{minipage}{\linewidth}
		\scriptsize
		\begin{align*}
			\hspace{0pt}
			&\hasrt{\es{\sigma}\sigma(x\mapsto\delta_x)}\\
	\models &\hasrt{((\fs{\sigma}\sigma(x\mapsto\delta_x))\star(\es{\sigma}\sigma(\top)))\otimes\top}\\
			&v\coloneqq[x]\tag{\ruleref{ax:read-scaffold},\ \ruleref{rule:join-true}}\\
			&\hasrt{((\fs{\sigma}\sigma(x\mapsto\delta_x\land v=\delta_x))\star(\es{\sigma}\sigma(\top)))\otimes\top}\\
	\models &\hasrt{\es{\sigma}\sigma(\top)}\\
	\models &\hasrt{\es{\sigma_1}\en{t_1}\es{\sigma_2}\en{t_2}\en{v_0}\en{n}\sigma_1(v=v_0\land t_1=n)\land\sigma_2(v=v_0\land t_2\not=n)}\\
			&t\coloneqq\mathrm{nonDet}()\tag{\ruleref{ax:havoc}}\\
			&\hasrt{\es{\sigma_1}\es{\sigma_2}\en{v_0}\en{n}\sigma_1(v=v_0\land t=n)\land\sigma_2(v=v_0\land t\not=n)}\\
			&z\coloneqq\mathrm{alloc}()\tag{\ruleref{ax:alloc}}\\
			&\hasrt{(\fs{\sigma}\sigma(z\mapsto\_))\star(\es{\sigma_1}\es{\sigma_2}\en{v_0}\en{n}\sigma_1(v=v_0\land t=n)\land\sigma_2(v=v_0\land t\not=n))}\\
			&[z]\coloneqq v+t\tag{\ruleref{ax:write}}\\
			&\hasrt{(\fs{\sigma}\sigma(z\mapsto v+t))\star(\es{\sigma_1}\es{\sigma_2}\en{v_0}\en{n}\sigma_1(v=v_0\land t=n)\land\sigma_2(v=v_0\land t\not=n))}\\
	\models &\hasrt{(\fs{\sigma}\sigma(z\mapsto\delta_z))\star(\es{\sigma_1}\es{\sigma_2}\en{n}\sigma_1(\delta_z=n)\land\sigma_2(\delta_z\neq n))}
		\end{align*}
		\caption{Proof that the program in black exhibits non-determinism.}
		\label{fig:non-determinism-example}
	\end{minipage}
\end{figure}

The next example (\figref{fig:non-determinism-example}) illustrates non-determinism, a $\exists\exists$-hyperproperty.
We begin with an initial set containing a state with $x$ allocated.
First, we read the value pointed to by $x$ into $v$, then allocate a new location $z$ and write $v+t$, where $t$ is chosen non-deterministically.
This example demonstrates that we can formally prove the existence of multiple distinct executions.

\subsection{Existence of Maximum}

\begin{figure}
	\begin{minipage}{\linewidth} 
		\tiny
		\begin{align*}
			\hspace{10pt}
			&\hasrt{(\es{\sigma_0}\fs{\sigma}\sigma(\ok:\top)} \\
			&\Alloc{x}; \tag{\ruleref{ax:alloc}}\\
			&\hasrt{(\fs{\sigma}\sigma(\ok:x\mapsto\_))\star(\es{\sigma_0}\fs{\sigma}\sigma(\ok:\top)} \\		
	\models &\hasrt{(\fs{\sigma}\fn{t}\sigma(\ok:t\not\le9)\lor\sigma(\ok:x\mapsto\_))\star(\es{\sigma_0}\en{t_0}\sigma_0(\ok:t_0\le9)\land\fs{\sigma}\fn{t}\sigma(ok:t\not\le9)\lor\en{n}\sigma_0(ok:t_0=n)\land\sigma(\ok:t\le n))} \\
			&\Havoc{t}; \tag{\ruleref{ax:havoc}}\\
			&\hasrt{(\fs{\sigma}\sigma(\ok:t\not\le9)\lor\sigma(\ok:x\mapsto\_))\star(\es{\sigma_0}\sigma_0(\ok:t\le9)\land\fs{\sigma}\sigma(\ok:t\not\le9)\lor\en{n}\sigma_0(\ok:t=n)\land\sigma(\ok:t\le n))} \\
			&\Assume{t\le9}; \tag{\ruleref{ax:assume}}\\
			&\hasrt{(\fs{\sigma}\sigma(\ok:x\mapsto\_))\star(\es{\sigma_0}\fs{\sigma}\en{n}\sigma_0(\ok:t=n)\land\sigma(\ok:t\le n))} \\
			&\Read{x}{t} \tag{\ruleref{ax:write}}\\
			&\hasrt{(\fs{\sigma}\sigma(\ok:x\mapsto t))\star(\es{\sigma_0}\fs{\sigma}\en{n}\sigma_0(\ok:t=n)\land\sigma(\ok:t\le n))} \\
		\end{align*}
		\caption{Proof that, when the program in black is executed from a non-empty $\ok$-only set of initial states, the set of final states (all with $x$ allocated) contains at least one with a maximal allocated value.}
		\label{fig:example-max}
	\end{minipage}%
\end{figure}

The next example (\figref{fig:example-max}) illustrates existence of a maximal pointed-to value, a $\exists\forall$-hyperproperty.
We begin with a non-empty set of initial states, allocate memory at $x$, and then non-deterministically choose a natural number $t\le9$ to write at $x$.
Since the chosen non-deterministic value is bounded by $9$ from above, then there exists a maximal value of the pointed-to value, namely $9$.

Note the explicit use of labels here. This is necessary because one might involuntary write $\lnot\sigma(\ok:t\le9)$, i.e., $\sigma(\ok:t\not\le9)\lor\sigma(\er:\top)\lor\sigma(\uk:\top)$;
however, the intention is to negate only the $\ok$ part, i.e., $\sigma(\ok:t\not\le9)$.

\end{document}